\documentclass[12pt]{article}
\usepackage[a4paper,margin=1in]{geometry}
\usepackage{amsmath,amsthm,mathtools}
\usepackage[T1]{fontenc}
\usepackage[utf8]{inputenc}
\usepackage{newtxtext,newtxmath}
\usepackage{graphicx}
\usepackage{subcaption}
\usepackage{bm}
\usepackage{enumitem,microtype,setspace}
\usepackage{booktabs,placeins,titlesec}
\usepackage[font={normalsize,doublespacing},labelfont=bf,skip=10pt]{caption}
\usepackage[authoryear,round]{natbib}
\usepackage[hidelinks]{hyperref}

\newcommand{\AuthorName}{%
	Ping Zhao\textsuperscript{1} \quad
	Dan Zhuang\textsuperscript{2,*}%
}

\newcommand{\AuthorAffiliation}{%
	\textsuperscript{1}School of Mathematical Sciences,\\
	Tiangong University, Tianjin 300071, China\\[0.4em]
	\textsuperscript{2}School of Mathematics and Statistics,\\
	Fujian Normal University, Fuzhou 350117, China\\[0.4em]
	\textsuperscript{*}Corresponding author%
}

\newcommand{\AuthorEmail}{zhgdan@fjnu.edu.cn}
\mathtoolsset{showonlyrefs}
\allowdisplaybreaks[4]
\titleformat{\section}{\normalfont\normalsize\bfseries}{\thesection}{1em}{}
\titleformat{\subsection}{\normalfont\normalsize\bfseries\itshape}{\thesubsection}{1em}{}
\titleformat{\subsubsection}{\normalfont\normalsize\itshape}{\thesubsubsection}{1em}{}
\makeatletter
\renewcommand{\abstract}{\if@twocolumn
  \section*{\abstractname}%
  \else\begin{center}\bfseries\abstractname\end{center}\quotation\fi}
\renewcommand{\endabstract}{\if@twocolumn\else\endquotation\fi}
\makeatother

\newtheorem{assumption}{Assumption}
\newtheorem{theorem}{Theorem}[section]
\newtheorem{proposition}[theorem]{Proposition}
\newtheorem{corollary}[theorem]{Corollary}
\newtheorem{lemma}{Lemma}[section]
\theoremstyle{remark}

\DeclareMathOperator{\tr}{tr}
\DeclareMathOperator{\Var}{Var}
\DeclareMathOperator{\Cov}{Cov}
\DeclareMathOperator{\diag}{diag}

\DeclareMathOperator{\argmax}{arg\,max}
\newcommand{\R}{\mathbb R}
\newcommand{\E}{\mathbb E}
\newcommand{\PP}{\mathbb P}
\newcommand{\I}{\mathbf I}
\newcommand{\one}{\bm 1}
\newcommand{\trans}{\mathsf T}
\newcommand{\dto}{\xrightarrow{d}}
\newcommand{\pto}{\xrightarrow{p}}
\newcommand{\aalpha}{\bm\alpha}
\newcommand{\ahat}{\widehat{\bm\alpha}}
\newcommand{\Sig}{\mathbf\Sigma_N}
\newcommand{\SSS}{\mathbf S_T}
\newcommand{\QQ}{\mathbf Q_T}
\newcommand{\calH}{\mathcal H}
\newcommand{\calL}{\mathcal L}
\newcommand{\Gam}{\mathbf\Gamma}
\newcommand{\ind}{\mathbf 1}

\newcommand{\supp}{\operatorname{supp}}
\newcommand{\ii}{\mathrm{i}}
\hypersetup{pdfauthor={\AuthorName},pdftitle={Cauchy-Aggregated Ridge Tests for High-Dimensional Factor Pricing Models},pdfsubject={Joint null and local-alternative theory for covariance-based regularized alpha tests}}
\begin{document}
\begin{center}
{\large\bfseries Cauchy-Aggregated Ridge Tests for \\
High-Dimensional Factor Pricing Models}\par
\AuthorName\par
\AuthorAffiliation\par
\ifx\AuthorEmail\empty\else\AuthorEmail\par\fi
\end{center}

\begin{abstract}
In high-dimensional factor pricing models, the power of regularized alpha tests depends on the ridge parameter, whose optimal value varies with the unknown direction of pricing errors. We address this tuning uncertainty by combining ridge-specific p-values using the Cauchy rule. The resulting test retains the least-squares alpha estimator and residual sample covariance and accommodates more assets than observations. We establish the joint Gaussian limits of the component statistics under the null and local alternatives and derive an explicit covariance formula across ridge parameters. These results characterize the combined test's asymptotic distribution and local power and justify its tail calibration. 
Simulations show empirical rejection rates close to the nominal levels and power gains over the fixed-ridge benchmark, with power approaching that of signal-informed ridge benchmarks across the simulated settings.
\end{abstract}

\clearpage
\section{Introduction}
Testing whether alphas are jointly zero is a central problem in the empirical evaluation of linear factor pricing models. An alpha measures the expected excess return that is not accounted for by the candidate factors. A joint test therefore evaluates the pricing restrictions across the entire collection of test assets rather than considering each asset separately. 
The classical Gibbons–Ross–Shanken (GRS) test \citep{GRS1989} measures the estimated alpha vector relative to the residual covariance matrix and adjusts for the estimation of factor loadings. Its implementation requires inversion of the residual covariance matrix and becomes problematic when the number of test assets approaches the sample size, as the covariance estimate becomes ill-conditioned and eventually singular.

A substantial literature develops high-dimensional alpha tests for different patterns of pricing errors and residual dependence. Quadratic-form and random-projection procedures provide feasible joint tests for large cross-sections \citep{Lan2018,PY2024}. Screening, maximum-type tests and adaptive signal selection improve sensitivity to sparse pricing errors \citep{Fan2015,Feng2022,XiaZhang2024}. Combination procedures pool complementary evidence from sum-type, maximum-type and higher-order-norm statistics to accommodate unknown sparsity \citep{Yu2024,ZhaoAdaptive2026,ZhaoFisher2026,ZhaoLq2026}. Distribution-free and spatial-sign methods address non-Gaussian returns through sign transformations, inverse-norm weighting and robust sum--max combinations \citep{Gungor2013,ZhaoNonparam2022,Liu2023,ZhaoRobust2026}. For serially dependent observations, \citet{MaDepend2025} develop sum-type, maximum-type and Cauchy-combined alpha tests.

A complementary approach is to regularize the covariance inverse while retaining its information about cross-sectional dependence. Regularized Hotelling tests were developed by \citet{Chen2011}.
\citet{Li2020ARHT} studied power-directed ridge selection and constructed a composite test by combining ridge statistics tailored to different local alternatives.
The general linear hypothesis theory of \citet{LAP2020} accommodates nuisance regressors and analytic transformations of the residual covariance spectrum. Building on this theory, \citet{CKMS2025} extend the GRS test to settings in which the number of assets can exceed the number of observations.

Ridge regularization stabilizes the inverse residual covariance matrix, while the centering and variance corrections account for high dimensionality. The ridge parameter determines the relative weights assigned to covariance directions and hence affects power. Smaller values place greater relative weight on low-variance directions, whereas larger values produce more uniform weighting. The preferred choice therefore depends on the alignment of pricing errors with the residual covariance eigenvectors.

The local-power criterion depends on the unknown alpha vector, or its covariance allocation under the random local alternatives of \citet{CKMS2025}. These authors propose a power-directed selection rule and report favorable performance for a standardized ridge near one in their simulations. 
We address uncertainty in the power-optimal ridge parameter by aggregating ridge-specific p-values over a prespecified grid using the Cauchy combination rule of \citet{LiuXie2020}. The resulting procedure does not require estimation of the unknown pricing-error allocation.
We derive joint Gaussian limits under the null and local alternatives, accounting for factor-loadings estimation and nonzero factor means, and obtain an explicit cross-ridge covariance. These results characterize the combined test's local power and establish relative size validity of its analytic calibration at small significance levels, without estimating dependence across ridge values. Simulations show power gains over the recommended fixed-ridge CKMS specification and performance close to signal-informed ridge benchmarks.

The remainder of the paper is organized as follows. Section~\ref{sec:method} presents the model, regularity conditions, single-ridge procedures and the proposed combination test. Section~\ref{sec:theory} establishes the joint asymptotic limits and size properties. Section~\ref{sec:simulation} reports the simulation designs and finite-sample results. Section~\ref{sec:conclusion} concludes. Auxiliary results and proofs are collected in the appendices.

\noindent\textit{Notation.} 
Let $T$ denote the number of time observations, and let $N$ and $p$ denote the numbers of test assets and factors, respectively, with $p$ fixed.
Vectors are columns and matrices are bold capitals. The symbols $\|\cdot\|$, $\|\cdot\|_\infty$ and $\|\cdot\|_{\mathrm{op}}$ denote the Euclidean, maximum-coordinate and spectral norms, respectively; $\tr(\cdot)$ denotes the trace. Normal distributions are denoted by $\mathcal{N}$. The standard normal distribution function is $\Phi$, its density is $\phi$, and $z_u=\Phi^{-1}(u)$. The significance level is denoted by $\eta$, and a point mass at $x$ by $\delta_x$. Constants denoted by $C$ are finite, may differ between displays, and do not depend on $N$ or $T$.

\section{From a single ridge parameter to a combined test}\label{sec:method}
\subsection{The model and factor adjustment}
For $t=1,\ldots,T$, let $\bm y_t\in\R^N$ collect the excess returns on $N$ test assets and let $\bm f_t\in\R^p$ collect the observed factors, where $p$ is fixed. We consider
\begin{equation*}
 \bm y_t=\aalpha+\mathbf B\bm f_t+\bm\varepsilon_t,
 \qquad \bm\varepsilon_t=\Sig^{1/2}\bm u_t,
 \qquad H_0:\aalpha=\bm0\quad\text{against}\quad H_1:\aalpha\ne\bm0.
\end{equation*}
Here $\mathbf B\in\R^{N\times p}$ is the factor-loading matrix and $\Sig$ is the residual covariance. The innovations are independent of the factor design, and inference is conditional on that design. The specification includes a single candidate portfolio ($p=1$) and a fixed number of factors, as in Appendix D of \citet{CKMS2025}. \footnote{References
	to specific results, section or equation numbers, and implementation details
	from \citet{CKMS2025} refer to the manuscript version dated February 27, 2025.}

Write $\bar{\bm y}=T^{-1}\sum_t\bm y_t$, $\bar{\bm f}=T^{-1}\sum_t\bm f_t$ and
\[
 \mathbf V_{f,T}=\frac1T\sum_{t=1}^T
 (\bm f_t-\bar{\bm f})(\bm f_t-\bar{\bm f})^{\trans}.
\]
The least-squares regression includes an intercept. Its slope and intercept estimates are
\begin{equation}\label{eq:OLS}
 \widehat{\mathbf B}=\left\{\frac1T\sum_{t=1}^T
 (\bm y_t-\bar{\bm y})(\bm f_t-\bar{\bm f})^{\trans}\right\}\mathbf V_{f,T}^{-1},
 \qquad
 \ahat=\bar{\bm y}-\widehat{\mathbf B}\bar{\bm f}.
\end{equation}
Let $\widehat{\bm a}_t=\bm y_t-\widehat{\mathbf B}\bm f_t$. The covariance estimate and factor-mean correction are
\begin{equation}\label{eq:cov}
 \SSS=\frac1T\sum_{t=1}^T
 (\widehat{\bm a}_t-\ahat)(\widehat{\bm a}_t-\ahat)^{\trans},
 \qquad d_T=1+\bar{\bm f}^{\trans}\mathbf V_{f,T}^{-1}\bar{\bm f}.
\end{equation}
The covariance normalization by $T$ follows \citet{CKMS2025}.

For the matrix representation, let $\mathbf Y=(\bm y_1,\ldots,\bm y_T)$ and $\mathbf E=(\bm\varepsilon_1,\ldots,\bm\varepsilon_T)$. Let $\mathbf F$ have rows $\bm f_t^{\trans}$, set $\mathbf F_c=\mathbf F-\one\bar{\bm f}^{\trans}$ and $\mathbf X=(\one,\mathbf F)$, and let $\bm e_1$ be the first coordinate vector in $\R^{p+1}$. Define
\begin{equation}\label{eq:projection}
 \begin{split}
 \mathbf M_X&=\I_T-\mathbf X(\mathbf X^{\trans}\mathbf X)^{-1}\mathbf X^{\trans},\qquad
 \bm v_T=\frac1T\{\one-\mathbf F_c\mathbf V_{f,T}^{-1}\bar{\bm f}\},\\
 \ahat&=\aalpha+\mathbf E\bm v_T,\qquad
 \SSS=\frac1T\mathbf E\mathbf M_X\mathbf E^{\trans},\qquad
 \|\bm v_T\|^2=\frac{d_T}{T}.
 \end{split}
\end{equation}
The projection eliminates $\aalpha$ and $\mathbf B$ from $\SSS$, while the factor design enters the alpha estimation error through $\bm v_T$.

\subsection{Regularity conditions}
The sampling and local-alternative framework follows \citet{CKMS2025}. Let $\tau_{r,N}$ be the eigenvalues of $\Sig$ and $H_N=N^{-1}\sum_{r=1}^N\delta_{\tau_{r,N}}$. For probability distributions $F_1,F_2$, define
\[
 W_1(F_1,F_2)=\sup_{\operatorname{Lip}(f)\le1}
 \left|\int f(x)\,dF_1(x)-\int f(x)\,dF_2(x)\right|,
\]
where the supremum is over real functions with Lipschitz constant at most one.

\begin{assumption}[Sampling and spectral stability]\label{ass:base}
The $u_{it}$ are i.i.d., with $\E u_{it}=0$, $\E u_{it}^2=1$ and $\E u_{it}^4<\infty$. Further,
\begin{enumerate}[label=(\roman*)]
\item $c_T=N/T\to c\in(0,\infty)$ and $\sqrt T\,|c_T-c|\to0$.
\item $\Sig\succeq0$, $\|\Sig\|_{\mathrm{op}}\le C_\Sigma$ and $\sqrt N\,W_1(H_N,H)\to0$, where $H$ has compact support in $(0,\infty)$.
\item The factor covariance is positive definite and uniformly bounded.
\end{enumerate}
\end{assumption}
Assumption~\ref{ass:base} permits dense residual covariance matrices and non-Gaussian innovations with finite fourth moments. Spectral stability concerns the empirical eigenvalue distribution; a uniform positive lower bound on individual eigenvalues is not required.

The regression design is also required to be bounded, fixed-dimensional and asymptotically full rank, with an estimable intercept. These conditions are stated in Appendix~\ref{app:design} and correspond to C5--C6 of \citet{LAP2020}, the regression theorem used by \citet{CKMS2025}. They supplement the factor-covariance condition in Assumption~\ref{ass:base}. Let $d=\lim_T d_T\in[1,\infty)$.

For the local alternatives, let $\mathbf\Omega_N$ be a deterministic nonnegative definite matrix and $\bm\gamma_N=(\gamma_1,\ldots,\gamma_N)^{\trans}$ an independent draw, independent of the innovations and design. For $s>0$, $m_\Omega(s)$ denotes a finite deterministic limit. The local signal is
\begin{equation}\label{eq:localalpha}
 \aalpha_T=T^{-1/4}N^{-1/2}\mathbf\Omega_N^{1/2}\bm\gamma_N.
\end{equation}
\begin{assumption}[Local alternatives]\label{ass:alt}
In \eqref{eq:localalpha}, the $\gamma_i$ are i.i.d., with $\E\gamma_i=0$, $\E\gamma_i^2=1$ and $\E\gamma_i^4<\infty$. For each $s>0$,
\begin{equation}\label{eq:originalsignal}
 \sqrt T\,\aalpha_T^{\trans}(\Sig+s\I_N)^{-1}\aalpha_T\pto m_\Omega(s),
 \qquad
 \frac1N\tr\{\mathbf\Omega_N(\Sig+s\I_N)^{-1}\}\to m_\Omega(s).
\end{equation}
\end{assumption}
The two limits in \eqref{eq:originalsignal} specify concentration of the signal quadratic form around its deterministic trace limit, as in equation (34) of \citet{CKMS2025}. Appendix~\ref{app:signal} derives the corresponding weighted spectral limit. The null corresponds to $\mathbf\Omega_N=\bm0$.

\subsection{Single-ridge asymptotic results}
For $\lambda>0$, let $\QQ(\lambda)=(\SSS+\lambda\I_N)^{-1}$. The regularized GRS statistic of \citet{CKMS2025} is
\begin{equation}\label{eq:W}
 W_T(\lambda)=\frac{\ahat^{\trans}\QQ(\lambda)\ahat}{d_T}.
\end{equation}
Define $\widehat m_T(\lambda)=N^{-1}\tr\QQ(\lambda)$ and
\begin{equation}\label{eq:feasible}
 \begin{split}
 \widehat\xi_T(\lambda)&=\{1-c_T+c_T\lambda\widehat m_T(\lambda)\}^{-1}-1,\\
 \widehat\Delta_T(\lambda)&=2\{\widehat\xi_T(\lambda)+\lambda\widehat\xi_T'(\lambda)\}
                         \{1+\widehat\xi_T(\lambda)\}^2,\\
 Z_T(\lambda)&=\frac{\sqrt T\{W_T(\lambda)-\widehat\xi_T(\lambda)\}}
                         {\sqrt{\widehat\Delta_T(\lambda)}}.
 \end{split}
\end{equation}
Primes denote derivatives with respect to the positive ridge argument $\lambda$.

For population quantities, let $b_\lambda$ be the unique solution in $(0,1)$ of
\begin{equation*}
 b_\lambda+c\int\frac{b_\lambda\tau}{\lambda+b_\lambda\tau}\,dH(\tau)=1.
\end{equation*}
Set
\begin{equation*}
 \xi_\lambda=b_\lambda^{-1}-1,
 \qquad
 \Delta_\lambda=2(\xi_\lambda+\lambda\xi_\lambda')(1+\xi_\lambda)^2,
 \qquad
 \calH_\Omega(\lambda)=\frac{1}{d b_\lambda}
                        m_\Omega(\lambda/b_\lambda).
\end{equation*}
The effective ridge is $\lambda/b_\lambda$, and the local signal functional is scaled by $(d b_\lambda)^{-1}$.

Write $\delta_\Omega(\lambda)=\calH_\Omega(\lambda)/\sqrt{\Delta_\lambda}$.
\begin{proposition}[Single-ridge RHT; \citealp{CKMS2025}]\label{prop:marginal}
For each fixed $\lambda>0$, in the maintained regression setting,
\begin{equation}\label{eq:marginalalt}
 Z_T(\lambda)\dto
 \begin{cases}
 \mathcal{N}(0,1),& H_0\text{ and Assumption~\ref{ass:base}},\\
 \mathcal{N}\{\delta_\Omega(\lambda),1\},&\text{Assumptions~\ref{ass:base}--\ref{ass:alt}}.
 \end{cases}
\end{equation}
\end{proposition}
Thus the test $Z_T(\lambda)>z_{1-\eta}$ has limiting local power
\begin{equation}\label{eq:singlepower}
 \beta_\lambda(\eta,\Omega)=1-\Phi\{z_{1-\eta}-\delta_\Omega(\lambda)\}.
\end{equation}
Proposition~\ref{prop:marginal} restates Theorems 5--7 of \citet{CKMS2025} in the present notation.

\subsection{The power-optimal ridge depends on alpha}
Fix a compact candidate interval $\calL=[\lambda_-,\lambda_+]\subset(0,\infty)$. Since \eqref{eq:singlepower} is increasing in $\delta_\Omega(\lambda)$, the population power-optimal ridge is
\begin{equation}\label{eq:oracle}
 \lambda^\star_\Omega\in\argmax_{\lambda\in\calL}
       \frac{\calH_\Omega(\lambda)}{\sqrt{\Delta_\lambda}}
 =\argmax_{\lambda\in\calL}
       \frac{m_\Omega(\lambda/b_\lambda)}{b_\lambda\sqrt{\Delta_\lambda}}.
\end{equation}
The factor $d^{-1}$ affects local power but cancels from the maximization because it is independent of $\lambda$. The same criterion applies to a finite candidate set in place of $\calL$.

Let $\bm\psi_{r,N}$ be an orthonormal eigenbasis of $\Sig$. The finite-dimensional signal term corresponding to \eqref{eq:oracle} is
\begin{equation}\label{eq:alphadependence}
 \frac{\sqrt T}{d_T\sqrt{\Delta_\lambda}}
 \aalpha_T^{\trans}(\lambda\I_N+b_\lambda\Sig)^{-1}\aalpha_T
 =\frac{\sqrt T}{d_T\sqrt{\Delta_\lambda}}
   \sum_{r=1}^N\frac{(\bm\psi_{r,N}^{\trans}\aalpha_T)^2}
                        {\lambda+b_\lambda\tau_{r,N}}.
\end{equation}
Consequently, changing the allocation of alpha across the covariance eigenvectors can change the maximizing ridge even when the total alpha magnitude is held fixed. Under the random alternatives, this allocation is represented by $\mathbf\Omega_N$ and $m_\Omega$. Neither is known in the testing problem.

The objective is continuously differentiable and attains its maximum on $\calL$. Every interior maximizer satisfies
\[
 2\calH_\Omega'(\lambda)\Delta_\lambda
   -\calH_\Omega(\lambda)\Delta_\lambda'=0.
\]
This first-order condition depends on both the residual covariance spectrum and the unknown signal allocation.

\citet[Section~3.3 and Appendix F]{CKMS2025} estimate the signal component of the power criterion by a leave-one-out construction. Using the factor-adjusted observations $\widehat{\bm a}_t$ from \eqref{eq:cov}, define
\[
 \widehat{\mathbf\Psi}_T=\frac1T\sum_{t=1}^T
       \widehat{\bm a}_t\widehat{\bm a}_t^{\trans},\qquad
 \widehat{\aalpha}_{-t}=\ahat-\frac{\widehat{\bm a}_t}{T},\qquad
 \widehat{\mathbf\Psi}_{-t}=\widehat{\mathbf\Psi}_T-
       \frac{\widehat{\bm a}_t\widehat{\bm a}_t^{\trans}}{T},
\]
\[
 \widehat{\mathbf S}_{-t}=\widehat{\mathbf\Psi}_{-t}
       -\widehat{\aalpha}_{-t}\widehat{\aalpha}_{-t}^{\trans},\qquad
 \mathbf Q_{-t}(\lambda)=(\widehat{\mathbf S}_{-t}+\lambda\I_N)^{-1}.
\]
The leave-one-out signal estimate pairs the omitted observation with the ridge-weighted mean formed from the remaining observations:
\[
 \widehat H_T^{\mathrm{LOO}}(\lambda)=\frac1T\sum_{t=1}^T
       \widehat{\aalpha}_{-t}^{\trans}\mathbf Q_{-t}(\lambda)
       \widehat{\bm a}_t.
\]
These quantities retain the denominator $T$, and the factor loadings are estimated once from the full sample.

For computation, let $\widehat{\mathbf A}=(\widehat{\bm a}_1,\ldots,\widehat{\bm a}_T)$ and set
\[
 \mathbf A_T=T^{-1}\widehat{\mathbf A}^{\trans}\widehat{\mathbf A},\qquad
 \widehat s_T^{\mathrm{raw}}=\frac{\tr(\mathbf A_T)}{N}
   =\frac{\tr(\SSS)+\|\ahat\|^2}{N},\qquad
 \mathbf G_T(\lambda)=(\lambda\I_T+\mathbf A_T)^{-1}\mathbf A_T.
\]
Write $w_t(\lambda)=\{\mathbf G_T(\lambda)\}_{tt}$ and $\bar w_t(\lambda)=\{\mathbf G_T(\lambda)\one\}_t$. Replacing the leave-one-out inverse's denominator by its full-sample approximation, their closed-form algorithm uses
\[
 \widetilde H_T(\lambda)=
 \frac{\displaystyle\sum_{t=1}^T
       \frac{\bar w_t(\lambda)-w_t(\lambda)}{1-w_t(\lambda)}}
      {\displaystyle1-T^{-1}\sum_{t=1}^T\bar w_t(\lambda)},
 \qquad
 \widehat\Pi_T(\lambda)=
       \frac{\widetilde H_T(\lambda)}{\sqrt{\widehat\Delta_T(\lambda)}}.
\]
The full-sample denominator approximation and the common factor $1/T$ distinguish the preceding leave-one-out average from this computational score; the common factor does not affect the maximizer. The variance estimate $\widehat\Delta_T$ is defined in \eqref{eq:feasible}. Their candidate multipliers and selected penalty are
\begin{equation}\label{eq:loo_rule}
 \begin{aligned}
 \mathcal A_{\mathrm{CKMS}}&=\{10^{-5},10^{-3},10^{-2},10^{-1},1,10,20,100\},\\
 \widehat a_{\mathrm L}&\in\argmax_{a\in\mathcal A_{\mathrm{CKMS}}}
       \widehat\Pi_T(a\widehat s_T^{\mathrm{raw}}),\qquad
 \widehat\lambda_{\mathrm L,T}=\widehat a_{\mathrm L}\widehat s_T^{\mathrm{raw}}.
 \end{aligned}
\end{equation}
The selected-ridge test then compares $Z_T(\widehat\lambda_{\mathrm L,T})$ with $z_{1-\eta}$.

This rule estimates the power criterion and selects a ridge using the same observations employed for testing. Under $H_0$, the population signal criterion is zero at every ridge, so its maximization does not identify a unique population optimum. Normal calibration at each prespecified ridge does not justify normal calibration after selecting a ridge from the data. In Section~\ref{sec:simulation}, the LOO rule implemented with the common grid and degrees-of-freedom correction fails to control size in three of the four designs, despite the reasonably accurate sizes of the fixed-ridge tests. This limitation motivates combining evidence across a prespecified grid. We retain the recommended fixed multiplier $a=1$ as CKMS and report the adaptive LOO rule separately.

\subsection{Cauchy combination across ridge values}
Choose a fixed number $J$ of distinct positive ridge values $\lambda_1,\ldots,\lambda_J$ and fixed weights $w_j>0$ with $\sum_jw_j=1$. Define
\begin{equation}\label{eq:Cauchy}
 \begin{split}
 p_{j,T}&=1-\Phi\{Z_T(\lambda_j)\},\\
 C_T&=\sum_{j=1}^Jw_j\tan\{\pi(1/2-p_{j,T})\},\\
 p_{C,T}&=\frac12-\frac1\pi\arctan(C_T).
 \end{split}
\end{equation}
The analytic rule rejects when $p_{C,T}\le\eta$, equivalently when $C_T\ge\cot(\pi\eta)$ for $0<\eta<1/2$. Equal weights provide a simple default. The choice of grid and weights is made before examining the marginal $p$-values.

All components share the regression estimates in \eqref{eq:OLS}--\eqref{eq:cov}. Their marginal normal-reference $p$-values are obtained from \eqref{eq:W}--\eqref{eq:feasible} and combined through \eqref{eq:Cauchy}. The analytic Cauchy calibration requires no estimate of the cross-ridge covariance. If a marginal variance or an estimated ridge scale is zero, set all $Z_T(\lambda_j)=0$, $C_T=0$ and $p_{C,T}=1/2$; this event has probability tending to zero under the stated conditions.

\section{Joint asymptotic theory and size}\label{sec:theory}
\subsection{Cross-ridge covariance and the null distribution}
For $\lambda,\mu>0$, define the deterministic kernel
\begin{equation}\label{eq:K}
 K(\lambda,\mu)=
 \begin{cases}
 2(1+\xi_\lambda)(1+\xi_\mu)
 \displaystyle\frac{\mu\xi_\mu-\lambda\xi_\lambda}{\mu-\lambda},&\lambda\ne\mu,\\[7pt]
 \Delta_\lambda,&\lambda=\mu.
 \end{cases}
\end{equation}
Write $K_{ij}=K(\lambda_i,\lambda_j)$, $\mathbf K=(K_{ij})_{i,j=1}^J$ and $\Gamma_{ij}=K_{ij}/\sqrt{\Delta_{\lambda_i}\Delta_{\lambda_j}}$. Let $\Gam=(\Gamma_{ij})_{i,j=1}^J$.

\begin{theorem}[Joint null limit]\label{thm:null}
Under Assumption~\ref{ass:base}, in the maintained regression setting,
\begin{equation}\label{eq:jointnull}
 \bigl(Z_T(\lambda_1),\ldots,Z_T(\lambda_J)\bigr)^{\trans}
 \dto \mathcal{N}_J(\bm0,\Gam)\qquad\text{under }H_0.
\end{equation}
\end{theorem}
The matrix $\Gam$ is positive definite on distinct ridge values, with $0<\Gamma_{ij}<1$ for $i\ne j$.
The shared alpha and covariance estimators induce dependence across ridge values. The joint limit follows from the general linear spectral theorem applied to finite linear combinations of resolvents, with the factor projection included in each combination.

\subsection{The joint law under local alternatives}
Write $\bm\delta_\Omega=(\delta_\Omega(\lambda_1),\ldots,\delta_\Omega(\lambda_J))^{\trans}$.
\begin{theorem}[Joint local-alternative limit]\label{thm:alternative}
Under Assumptions~\ref{ass:base}--\ref{ass:alt}, in the maintained regression setting,
\begin{equation}\label{eq:jointalt}
 \bigl(Z_T(\lambda_1),\ldots,Z_T(\lambda_J)\bigr)^{\trans}
 \dto \mathcal{N}_J(\bm\delta_\Omega,\Gam).
\end{equation}
\end{theorem}
Local alternatives shift the means while leaving the first-order cross-ridge covariance unchanged. After factor-loadings estimation, a nonzero factor mean reduces each local noncentrality by the common factor $d^{-1}$ and leaves $\Gam$ unchanged.

Let $\mu_1=\int\tau\,dH(\tau)>0$ and fix distinct positive constants $a_1,\ldots,a_J$.
\begin{corollary}[A common estimated ridge scale]\label{cor:scale}
Theorems~\ref{thm:null}--\ref{thm:alternative} remain valid for
\[
 \widehat s_T^{\mathrm{raw}}=\frac{\tr(\SSS)+\|\ahat\|^2}{N},
 \qquad \lambda_{j,T}=a_j\widehat s_T^{\mathrm{raw}},
\]
with population ridge $\lambda_j=a_j\mu_1$.
\end{corollary}
The scale $\widehat s_T^{\mathrm{raw}}=(NT)^{-1}\sum_{t=1}^T\|\widehat{\bm a}_t\|^2$ is the average second moment of the factor-adjusted returns.

\subsection{The distribution of the Cauchy statistic}
For $z\in\R$, let $g(z)=\tan\{\pi(\Phi(z)-1/2)\}$. For a vector $\bm\delta\in\R^J$ and a correlation matrix $\mathbf R$, define the random variable
\begin{equation*}
 \mathscr C(\bm\delta,\mathbf R)=\sum_{j=1}^Jw_jg(G_j+\delta_j),
 \qquad \bm G\sim \mathcal{N}_J(\bm0,\mathbf R).
\end{equation*}
The definition also applies when $\mathbf R$ is singular.

\begin{corollary}[Limit of the combined statistic]\label{cor:Climit}
Under the respective conditions of Theorems~\ref{thm:null} and \ref{thm:alternative},
\[
 C_T\dto\mathscr C(\bm0,\Gam)\quad(H_0),\qquad
 C_T\dto\mathscr C(\bm\delta_\Omega,\Gam)\quad(H_1).
\]
\end{corollary}
The limits have continuous distribution functions. At a fixed $\eta\in(0,1/2)$, the limiting local rejection probability is
\begin{equation}\label{eq:Cpower}
 \beta_C(\eta,\Omega)=\PP\{\mathscr C(\bm\delta_\Omega,\Gam)>\cot(\pi\eta)\}.
\end{equation}
The components of $\bm\delta_\Omega$ are nonnegative and $g$ is increasing. Under a common Gaussian coupling, the limiting statistic under the alternative therefore dominates its null counterpart pointwise.

\subsection{Size of the analytic Cauchy rule}
The analytic rule in \eqref{eq:Cauchy} satisfies the following relative size guarantee under cross-ridge dependence.
\begin{theorem}[Relative small-level size validity]\label{thm:size}
Under the conditions of Theorem~\ref{thm:null},
\begin{equation}\label{eq:size}
 \lim_{\eta\downarrow0}\lim_{T\to\infty}
 \frac{\PP_0(p_{C,T}\le\eta)}{\eta}=1.
\end{equation}
\end{theorem}
More precisely, for $J\ge2$ define $\rho_* =\max_{i\ne j}\Gamma_{ij}<1$ and $\kappa=(1-\rho_*)/4$. For some $C<\infty$ and all sufficiently small $\eta>0$,
\begin{equation}\label{eq:sizebound}
 \left|\lim_{T\to\infty}\frac{\PP_0(p_{C,T}\le\eta)}{\eta}-1\right|
 \le C\eta^\kappa\{1+|\log\eta|\}^2.
\end{equation}
For $J=1$, the inner limit is one at every fixed $\eta\in(0,1)$.
Equation~\eqref{eq:size} is a sequential asymptotic result: $N,T\to\infty$ at a fixed significance level, followed by $\eta\downarrow0$. A simultaneous sequence $\eta_T\downarrow0$ additionally requires control of the Gaussian approximation error relative to $\eta_T$. 

\section{Simulation studies}\label{sec:simulation}
%The simulations use the Gaussian single-factor design of \citet{CKMS2025}, with $p=1$, $\mathbf B=\bm\beta$ and $\Sig=\diag(\tau_1,\ldots,\tau_N)$. The components are generated independently as

The simulations follow the single-factor Gaussian design of
\citet{CKMS2025}, with $p=1$, $\mathbf B=\bm\beta$, and
$\Sig=\diag(\tau_1,\ldots,\tau_N)$. 
The components are generated independently as
\[
 f_t\sim \mathcal{N}(0.07,0.14^2),\qquad
 \bm u_t\sim \mathcal{N}_N(\bm0,\I_N),\qquad
 \tau_i\sim\operatorname{Unif}[1,2],\qquad
 \beta_i\sim\operatorname{Unif}[-3,3].
\]
The factors and innovations are i.i.d. over time. The eigenvalues and loadings are redrawn across replications and held fixed within each replication.

Under $H_0$, we set $\aalpha=\bm0$. 
 Under $H_1$, we draw
$\gamma_i\stackrel{\mathrm{i.i.d.}}{\sim}\mathcal{N}(0,1)$ and define $\widetilde\alpha_i=T^{-1/4}N^{-1/2}(1+\tau_i/2)^{1/2}\gamma_i$.
This construction yields the allocation
$\mathbf\Omega_N=\I_N+\Sig/2$ in \eqref{eq:localalpha}.
%For signal strength $q>0$, use
For a signal strength parameter $q>0$, we set
\begin{equation}\label{eq:sim_signal}
 \aalpha(q)=
 \left\{\frac{q d_T}{T\widetilde{\aalpha}^{\trans}
                          \Sig^{-1}\widetilde{\aalpha}}\right\}^{1/2}
          \widetilde{\aalpha},
 \qquad
 \frac{T\aalpha(q)^{\trans}\Sig^{-1}\aalpha(q)}{d_T}=q.
\end{equation}
The designs summarized in Table~\ref{tab:sim_design} use $2{,}000$
replications under the null and $1{,}000$ replications for each
alternative configuration, with $\eta=0.05$. All methods are evaluated
on the same simulated samples within each replication.

\begin{table}[htbp]
	\centering
	\caption{Simulation designs and theoretically optimal ridge penalties.}
	\label{tab:sim_design}
	\begin{tabular}{rrrc}
		\toprule
		$T$ & $N$ & $\lambda_{\mathrm F,T}^{\star}$ & Signal strength $q$\\
		\midrule
		100 & 200 & 52.1848  & $10,20,\ldots,90,110$\\
		100 & 500 & 131.6813 & $16,32,\ldots,160$\\
		200 & 200 & 25.6823  & $10,20,\ldots,90,110$\\
		200 & 500 & 65.4254  & $20,35,\ldots,140,160$\\
		\bottomrule
	\end{tabular}
\end{table}

All methods use the least-squares estimates and covariance $\SSS$ defined in 
 Section~\ref{sec:method}, with the marginal degrees-of-freedom correction $n=T-2$ detailed in Appendix~\ref{app:sim_df}. Following Corollary~\ref{cor:scale}, we set
\begin{equation}\label{eq:sim_scale}
 \widehat s_T^{\mathrm{raw}}=\frac{\tr(\SSS)+\|\ahat\|^2}{N},
 \qquad \lambda_{j,T}=a_j\widehat s_T^{\mathrm{raw}},
 \qquad \mathcal A=\{0.1,1,10,20,100\}.
\end{equation}
CKMS denotes the recommended fixed 
multiplier specification with $a=1$,
whereas the proposed Cauchy procedure combines the ridge-specific \(p\)-values corresponding to the five multipliers in \(\mathcal A\) using equal weights and the analytic calibration in \eqref{eq:Cauchy}.

LOO denotes the power-directed leave-one-out selector of \citet[Section~3.3]{CKMS2025}, evaluated over the same set $\mathcal A$. 
It selects a ridge separately in each sample by maximizing the estimated signal criterion divided by its estimated standard deviation, and applies the normal critical value to the selected ridge statistic. Both selection and testing use the common degrees-of-freedom correction. 
Appendix~\ref{app:sim_loo} states the implemented criterion. %Thus 
CKMS is the fixed-ridge recommendation used in their main simulations, while LOO represents their feasible data-driven selection
rule under the present grid and normalization.

Two benchmarks assess the power attainable with additional information about the alternative. The \emph{empirical best lambda} selects, for each $(T,N,q)$, the %single
multiplier in $\mathcal A$ giving the highest rejection frequency across the alternative replications. The \emph{formula optimum} maximizes the population criterion \eqref{eq:oracle} under $H=\operatorname{Unif}[1,2]$ and $\mathbf\Omega_N=\I_N+\Sig/2$. %Its penalties,
The resulting ridge penalties,
reported in Table~\ref{tab:sim_design}, are fixed for 
each design; 
Appendix~\ref{app:sim_formula} provides the corresponding calculation
and covariance-scale conversion.

Table~\ref{tab:sim_size} reports empirical size at the nominal \(5\%\) level. The CKMS benchmark, Cauchy aggregation, and the formula optimum all exhibit reasonable size control, with only mild over-rejection. The individual prespecified ridge tests show similar behavior across the candidate multipliers. In contrast, LOO is substantially oversized in three of the four designs. The pointwise size behavior of the fixed-ridge tests is therefore not preserved by the data-dependent LOO selection step. Cauchy aggregation maintains rejection frequencies comparable to the fixed-ridge procedures while combining evidence over the full grid. The empirical-grid oracle is not reported under the null because its multiplier is defined using rejection frequencies under the alternative.

\begin{table}[htbp]
\centering
\caption{Empirical size (\%) at the nominal $5\%$ level based on $2{,}000$ null replications.}
\label{tab:sim_size}
\begingroup\normalsize
\begin{tabular}{lrrrr}
\toprule
 & \multicolumn{2}{c}{$T=100$} & \multicolumn{2}{c}{$T=200$}\\
\cmidrule(lr){2-3}\cmidrule(lr){4-5}
 & $N=200$ & $N=500$ & $N=200$ & $N=500$\\
\midrule
\multicolumn{5}{l}{\emph{Panel A: Fixed ridge multipliers}}\\
$a=0.1$ & 6.30 & 5.75 & 6.60 & 6.65\\
$a=1$   & 5.95 & 5.95 & 6.40 & 6.20\\
$a=10$  & 6.20 & 5.70 & 5.75 & 6.05\\
$a=20$  & 5.90 & 5.45 & 5.65 & 6.00\\
$a=100$ & 5.80 & 5.30 & 5.80 & 6.05\\
\midrule
\multicolumn{5}{l}{\emph{Panel B: Testing procedures}}\\
CKMS ($a=1$) & 5.95 & 5.95 & 6.40 & 6.20\\
LOO & 9.30 & 5.75 & 10.70 & 9.00\\
Cauchy & 6.30 & 6.00 & 6.25 & 6.35\\
Formula optimum & 6.00 & 5.30 & 5.65 & 6.00\\
\bottomrule
\end{tabular}
\endgroup
\end{table}

Figures~\ref{fig:sim_power_T100_N200}--\ref{fig:sim_power_T200_N500} report power against the signal strength $q$ in \eqref{eq:sim_signal}, using the same simulated samples for all procedures. Cauchy aggregation has higher power than fixed-ridge CKMS throughout all four designs. The gains are most pronounced at intermediate signals, where improved sensitivity is particularly useful for detecting pricing errors that are not already large enough for all procedures to reject. This advantage is accompanied by size behavior comparable to CKMS, supporting aggregation as an effective way to improve power over the recommended fixed-ridge specification.

Cauchy aggregation also closely tracks both signal-informed benchmarks across the designs and signal strengths. 
The empirical-grid oracle selects its multiplier using rejection frequencies across alternative replications, whereas the formula optimum uses the generating signal allocation. In contrast, Cauchy aggregation uses only the observed sample and a prespecified ridge grid. It therefore attains power close to these benchmarks without knowledge of the direction of the pricing errors or estimation of the power-maximizing penalty. In these designs, aggregation thereby reduces the power loss associated with committing to a single ridge while retaining size behavior comparable to the prespecified ridge tests.

LOO is oversized in three of the four designs, so its higher rejection frequencies cannot be interpreted as power gains at a common size. In the remaining design, its size is satisfactory but its power is lower than that of the other procedures.
\FloatBarrier

\begin{figure}[htbp]
	\centering
	
	\begin{subfigure}[t]{0.48\textwidth}
		\centering
		\includegraphics[width=\textwidth]
		{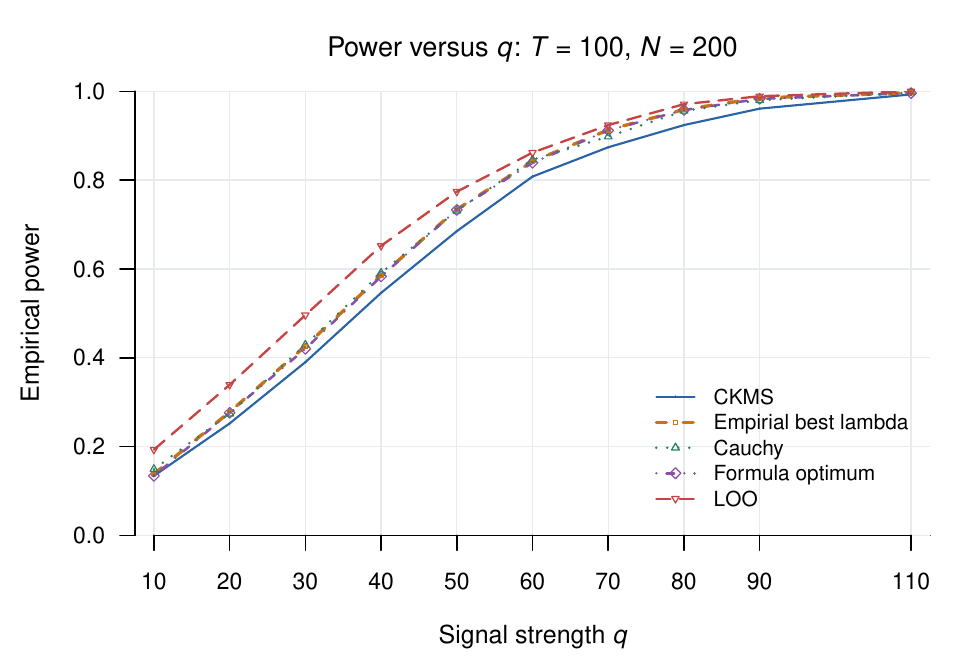}
		\caption{$T=100$, $N=200$.}
		\label{fig:sim_power_T100_N200}
	\end{subfigure}
	\hfill
	\begin{subfigure}[t]{0.48\textwidth}
		\centering
		\includegraphics[width=\textwidth]
		{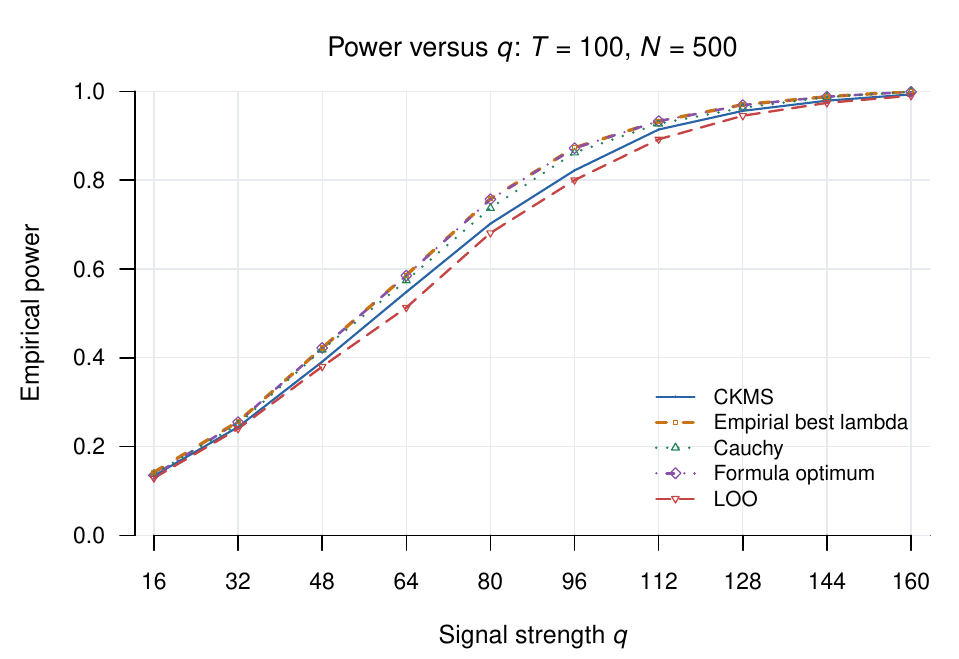}
		\caption{$T=100$, $N=500$.}
		\label{fig:sim_power_T100_N500}
	\end{subfigure}
	
	\vspace{0.8em}
	
	\begin{subfigure}[t]{0.48\textwidth}
		\centering
		\includegraphics[width=\textwidth]
		{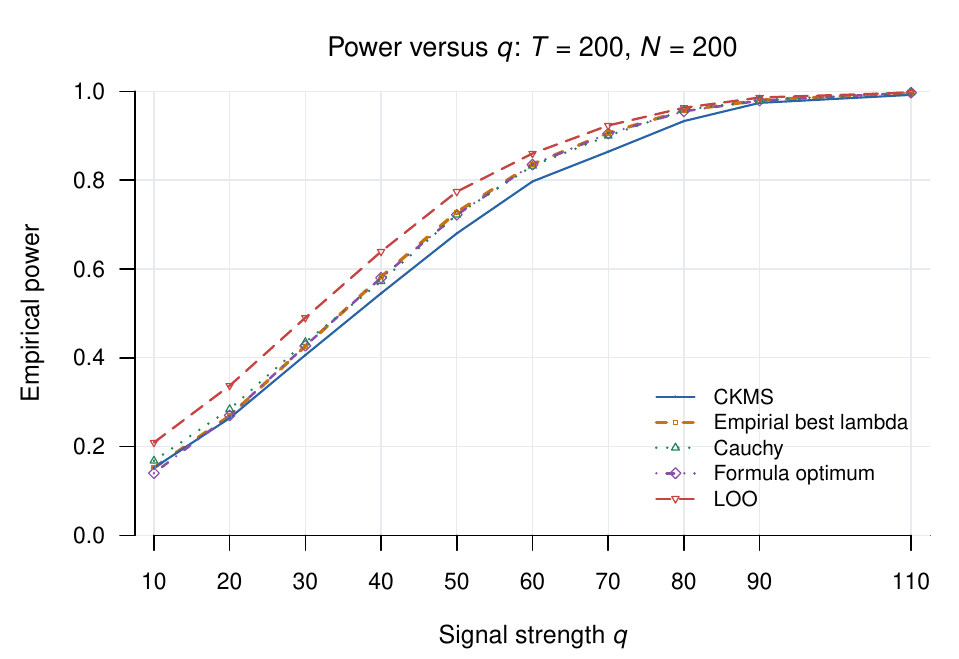}
		\caption{$T=200$, $N=200$.}
		\label{fig:sim_power_T200_N200}
	\end{subfigure}
	\hfill
	\begin{subfigure}[t]{0.48\textwidth}
		\centering
		\includegraphics[width=\textwidth]
		{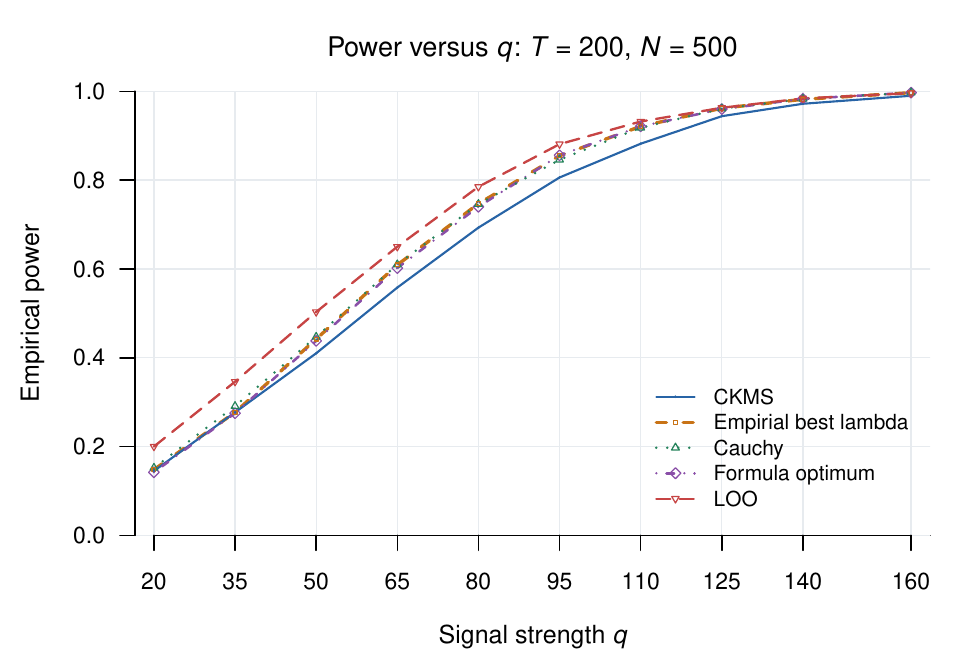}
		\caption{$T=200$, $N=500$.}
		\label{fig:sim_power_T200_N500}
	\end{subfigure}
	
	\caption{Empirical power against the signal strength $q$
		at the nominal level $\eta=0.05$ across the four simulation designs.}
	\label{fig:sim_power}
\end{figure}

\FloatBarrier
\section{Conclusion}\label{sec:conclusion}
This paper develops a Cauchy combination of regularized alpha tests over a finite ridge grid. Aggregation accommodates uncertainty about the pricing-error allocation that determines the power-optimal ridge. The joint null and local-alternative limits account for factor-loadings estimation and yield an explicit cross-ridge covariance. The analytic calibration has relative size validity at small significance levels and requires no estimation of the cross-ridge correlation matrix. In the Gaussian simulations, aggregation improves on the recommended fixed-ridge CKMS specification and attains power close to the signal-informed benchmarks. Extensions to serially dependent innovations and growing ridge grids warrant further study.
\FloatBarrier

\par\medskip
\noindent\textit{Funding.}
Ping Zhao was supported by the National Natural Science Foundation of China
(Grant No.~12601527), and Dan Zhuang was supported by the National Natural
Science Foundation of China (Grant No.~12401380).

\par\medskip
\noindent\textit{Disclosure statement.}
The authors report no competing interests to declare.

\par\medskip
\noindent\textit{Data availability.}
The data in this study are generated by Monte Carlo simulation. The accompanying
source archive includes the rejection-frequency summaries underlying the tables
and figures.

\par\medskip
\noindent\textit{Declaration of generative AI use.}
OpenAI ChatGPT (GPT-6) was used for manuscript organization, language editing, and
\LaTeX{} formatting. All AI-assisted content was reviewed and verified by the
authors, who reviewed the applicable terms of use and take full responsibility
for the originality, accuracy, integrity, and final content of the manuscript.

\begingroup
\doublespacing
\setlength{\bibsep}{0pt}
\bibliographystyle{apalike}
\bibliography{references}
\endgroup
\clearpage
\appendix
\section{Preliminaries for the joint ridge limits}\label{app:prelim}
\subsection{Regression design and factor adjustment}\label{app:design}
The regression-design conditions C5--C6 of \citet{LAP2020}, used in Appendix D of \citet{CKMS2025}, take the form
\begin{equation}\label{eq:designscope}
 p\text{ fixed},\qquad \sup_{T,t}\|\bm x_t\|<\infty,\qquad
 \mathbf D_T:=T^{-1}\mathbf X^{\trans}\mathbf X\longrightarrow\mathbf D\succ0,
 \qquad \bm x_t=(1,\bm f_t^{\trans})^{\trans}.
\end{equation}
The intercept is estimable since $\bm e_1^{\trans}\mathbf D^{-1}\bm e_1>0$.

For random factors, the following calculations are conditional on an independent design in \eqref{eq:designscope}. The innovation law is unchanged by that conditioning.

\begin{lemma}[Projection identities]\label{lem:projection}
Equations~\eqref{eq:projection} hold, and
\begin{equation*}
 \max_t|v_{t,T}|\le C/T,\qquad
 \max_t\{1-(\mathbf M_X)_{tt}\}\le C/T,\qquad
 \Cov(\ahat-\aalpha\mid\mathbf X)=\frac{d_T}{T}\Sig.
\end{equation*}
\end{lemma}
\begin{proof}
Let $\mathbf X_c=(\one,\mathbf F_c)$ and $\mathbf A_f=\left(\begin{smallmatrix}1&\bar{\bm f}^{\trans}\\ \bm0&\I_p\end{smallmatrix}\right)$. Centering gives
\[
 \mathbf X=\mathbf X_c\mathbf A_f,\qquad
 \mathbf X_c^{\trans}\mathbf X_c=
 T\begin{pmatrix}1&\bm0^{\trans}\\\bm0&\mathbf V_{f,T}\end{pmatrix},\qquad
 \mathbf A_f^{-1}=\begin{pmatrix}1&-\bar{\bm f}^{\trans}\\\bm0&\I_p\end{pmatrix}.
\]
Consequently,
\[
 (\mathbf X^{\trans}\mathbf X)^{-1}
 =\frac1T\begin{pmatrix}
 1+\bar{\bm f}^{\trans}\mathbf V_{f,T}^{-1}\bar{\bm f}
 &-\bar{\bm f}^{\trans}\mathbf V_{f,T}^{-1}\\
 -\mathbf V_{f,T}^{-1}\bar{\bm f}&\mathbf V_{f,T}^{-1}
 \end{pmatrix},
 \qquad
 \mathbf X(\mathbf X^{\trans}\mathbf X)^{-1}\bm e_1=\bm v_T.
\]
The least-squares normal equations and the matrix model imply
\[
 \begin{aligned}
 \mathbf X^{\trans}\bm v_T&=\bm e_1,
 &\one^{\trans}\bm v_T&=1,
 &\mathbf F^{\trans}\bm v_T&=\bm0,\\
 \ahat&=\mathbf Y\bm v_T=\aalpha+\mathbf E\bm v_T,
 &\mathbf Y\mathbf M_X&=\mathbf E\mathbf M_X,
 &\mathbf M_X\bm v_T&=\bm0.
 \end{aligned}
\]
The projection has the explicit form
\[
 \I_T-\mathbf M_X
 =\frac{\one\one^{\trans}}T+
   \frac1T\mathbf F_c\mathbf V_{f,T}^{-1}\mathbf F_c^{\trans},\qquad
 \mathbf M_X^{\trans}=\mathbf M_X,\quad \mathbf M_X^2=\mathbf M_X.
\]
Hence
\[
 \SSS=T^{-1}(\mathbf Y\mathbf M_X)(\mathbf Y\mathbf M_X)^{\trans}
       =T^{-1}\mathbf E\mathbf M_X\mathbf E^{\trans},\qquad
 \|\bm v_T\|^2=\bm e_1^{\trans}(\mathbf X^{\trans}\mathbf X)^{-1}\bm e_1=d_T/T.
\]
Independence across dates now gives
\[
 \Cov(\mathbf E\bm v_T\mid\mathbf X)
 =\sum_{t,s}v_{t,T}v_{s,T}\E(\bm\varepsilon_t\bm\varepsilon_s^{\trans})
 =\sum_t v_{t,T}^2\Sig=\frac{d_T}{T}\Sig.
\]
Since $\|\mathbf D_T^{-1}\|_{\mathrm{op}}\le C$ for large $T$,
\[
 |v_{t,T}|\le T^{-1}\|\bm x_t\|\|\mathbf D_T^{-1}\|_{\mathrm{op}},\qquad
 1-(\mathbf M_X)_{tt}
 \le T^{-1}\|\bm x_t\|^2\|\mathbf D_T^{-1}\|_{\mathrm{op}}.
\]
Finally, the inverse identity yields the explicit design error
\begin{align*}
 d&=\bm e_1^{\trans}\mathbf D^{-1}\bm e_1,\qquad
 \mathbf D_T^{-1}-\mathbf D^{-1}
   =\mathbf D_T^{-1}(\mathbf D-\mathbf D_T)\mathbf D^{-1},\\
 |d_T-d|&\le C\|\mathbf D_T-\mathbf D\|_{\mathrm{op}}\longrightarrow0.
 \end{align*}
These covariance calculations use independence of the innovation vectors across dates, without requiring independence of $\ahat$ and $\SSS$.
\end{proof}

\subsection{Quadratic moments and the common ridge scale}
\begin{lemma}[Quadratic-form covariance]\label{lem:quadratic}
Let $\bm z$ have independent standardized entries with common fourth cumulant $\chi_4=\E z_1^4-3$. For deterministic symmetric matrices $\mathbf A,\mathbf B$,
\begin{equation}\label{eq:quadraticmoment}
 \Cov(\bm z^{\trans}\mathbf A\bm z,\bm z^{\trans}\mathbf B\bm z)
 =2\tr(\mathbf A\mathbf B)+\chi_4\sum_a A_{aa}B_{aa}.
\end{equation}
\end{lemma}
\begin{proof}
Independence and $\E z_a=0$ give
\[
 \E(z_az_bz_cz_d)
 =\ind_{\{a=b\}}\ind_{\{c=d\}}
 +\ind_{\{a=c\}}\ind_{\{b=d\}}
 +\ind_{\{a=d\}}\ind_{\{b=c\}}
 +\chi_4\ind_{\{a=b=c=d\}}.
\]
Indeed, an index occurring once makes the expectation zero; paired indices contribute one, and four equal indices contribute $3+\chi_4$.
\begin{align*}
 \E(\bm z^{\trans}\mathbf A\bm z)&=\sum_aA_{aa}=\tr\mathbf A,\\
 \E\{(\bm z^{\trans}\mathbf A\bm z)(\bm z^{\trans}\mathbf B\bm z)\}
 &=\tr\mathbf A\tr\mathbf B
   +\sum_{a,b}A_{ab}B_{ab}+\sum_{a,b}A_{ab}B_{ba}
   +\chi_4\sum_aA_{aa}B_{aa}.
 \end{align*}
Subtracting the product of the means and using symmetry proves \eqref{eq:quadraticmoment}. In particular,
\[
 \Var(\bm z^{\trans}\mathbf A\bm z)
 \le(2+|\chi_4|)\tr(\mathbf A^2),
 \qquad \sum_aA_{aa}^2\le\sum_{a,b}A_{ab}^2=\tr(\mathbf A^2).\qedhere
\]
\end{proof}

Put $n=T-p-1$, $\widehat s_T=N^{-1}\tr\SSS$ and $\mu_1=\int\tau\,dH(\tau)$. The deterministic and centered random scale errors are
\[
 r_{s,T}=\frac n{TN}\tr\Sig-\mu_1,\qquad
 R_{s,T}=\widehat s_T-\frac n{TN}\tr\Sig.
\]
\begin{lemma}[Trace and alpha bounds]\label{lem:scale}
Under Assumption~\ref{ass:base},
\begin{equation*}
 \widehat s_T-\mu_1=R_{s,T}+r_{s,T},\quad
 \E(R_{s,T}^2\mid\mathbf X)\le\frac C{NT},\quad
 |r_{s,T}|\le C/T+W_1(H_N,H).
\end{equation*}
Under Assumption~\ref{ass:alt}, $\E\|\aalpha_T\|^2\le CT^{-1/2}$ and $\E\|\ahat\|^2\le C$.
\end{lemma}
\begin{proof}
Let $\bm z=(\bm u_1^{\trans},\ldots,\bm u_T^{\trans})^{\trans}$. The trace identity and Lemma~\ref{lem:projection} give
\[
 NT\widehat s_T=\bm z^{\trans}(\mathbf M_X\otimes\Sig)\bm z,
 \qquad
 \E(\widehat s_T\mid\mathbf X)=\frac{\tr\mathbf M_X\tr\Sig}{NT}
 =\frac n{NT}\tr\Sig.
\]
Using Lemma~\ref{lem:quadratic}, $\mathbf M_X^2=\mathbf M_X$ and $\|\Sig\|_{\mathrm{op}}\le C_\Sigma$,
\begin{align*}
 \Var(NT\widehat s_T\mid\mathbf X)
 &=2n\tr(\Sig^2)+\chi_4\sum_t(\mathbf M_X)_{tt}^2\sum_i\Sigma_{ii}^2\\
 &\le(2+|\chi_4|)n\tr(\Sig^2)\le CNT.
 \end{align*}
The deterministic part satisfies
\[
 |r_{s,T}|\le\frac{p+1}{T}\frac{\tr\Sig}{N}
       +\left|\int\tau\,dH_N(\tau)-\int\tau\,dH(\tau)\right|
 \le C/T+W_1(H_N,H).
\]
In particular, for every $x>0$,
\begin{equation*}
 \PP\{|R_{s,T}|>x\mid\mathbf X\}\le\frac C{NTx^2},\qquad
 \widehat s_T-\mu_1=O_p(T^{-1})+O\{W_1(H_N,H)\}.
\end{equation*}
Fix $s_0>0$. The trace limit in Assumption~\ref{ass:alt} implies
\[
 \frac{\tr\mathbf\Omega_N}{N}
 \le\frac{C_\Sigma+s_0}{N}\tr\{\mathbf\Omega_N(\Sig+s_0\I_N)^{-1}\}\le C.
\]
Independence of the alpha draw and innovations then yields
\begin{align*}
 \E\|\aalpha_T\|^2&=T^{-1/2}N^{-1}\tr\mathbf\Omega_N\le CT^{-1/2},\\
 \E\|\ahat\|^2&=\E\|\aalpha_T\|^2+\frac{d_T}{T}\tr\Sig\le C,\\
 \PP\{\|\aalpha_T\|>xT^{-1/4}\}&\le C/x^2,
 \qquad
 \PP\{N^{-1}\|\ahat\|^2>x/T\}\le C/x.
 \end{align*}
Thus $\|\aalpha_T\|=O_p(T^{-1/4})$, $\|\ahat\|=O_p(1)$ and $N^{-1}\|\ahat\|^2=O_p(T^{-1})$.
\end{proof}

\subsection{Fixed-point and cross-ridge identities}\label{app:kernel}
For $\lambda>0$, define
\[
 D_\lambda(\tau)=\frac1{\lambda+b_\lambda\tau},\qquad
 B_{1,\lambda}=\int\tau D_\lambda(\tau)^2\,dH(\tau),\qquad
 A_{\lambda\mu}=\int\tau^2D_\lambda(\tau)D_\mu(\tau)\,dH(\tau).
\]
\begin{lemma}[Fixed-point identities]\label{lem:fixedpoint}
The solution $b_\lambda$ is unique and smooth for $\lambda>0$, and
\begin{equation}\label{eq:bprime}
 b_\lambda'=\frac{cb_\lambda B_{1,\lambda}}{1+c\lambda B_{1,\lambda}},\qquad
 \Delta_\lambda=\frac{2cA_{\lambda\lambda}}{1-cb_\lambda^2A_{\lambda\lambda}}>0.
\end{equation}
\end{lemma}
\begin{proof}
For fixed $\lambda>0$, set $F_\lambda(b)=b+c\int b\tau/(\lambda+b\tau)\,dH(\tau)$. Then
\[
 F_\lambda(0)=0,\quad F_\lambda(1)>1,\quad
 F_\lambda'(b)=1+c\int\frac{\lambda\tau}{(\lambda+b\tau)^2}\,dH(\tau)>0.
\]
There is exactly one $b_\lambda\in(0,1)$. On a compact interval $\lambda\in[\ell,u]\subset(0,\infty)$,
\[
 1=F_\lambda(b_\lambda)
 \le b_\lambda(1+cC_\Sigma/\ell),\qquad
 b_\lambda\ge(1+cC_\Sigma/\ell)^{-1}>0.
\]
Differentiating $F_\lambda(b_\lambda)=1$ gives
\[
 b_\lambda'+c\int
 \frac{\tau(\lambda b_\lambda'-b_\lambda)}{(\lambda+b_\lambda\tau)^2}\,dH(\tau)=0,
 \qquad
 b_\lambda'(1+c\lambda B_{1,\lambda})=cb_\lambda B_{1,\lambda}.
\]
The denominator is positive. Repeated differentiation is legitimate on $[\ell,u]$, because every denominator is at least $\ell$ and the support of $H$ is bounded.
\[
 \int\tau D_\lambda(\tau)\,dH(\tau)
 =\lambda B_{1,\lambda}+b_\lambda A_{\lambda\lambda},
 \qquad
 b_\lambda(1+c\lambda B_{1,\lambda})
 =1-cb_\lambda^2A_{\lambda\lambda}>0.
\]
Since $\xi_\lambda=c\int\tau D_\lambda\,dH$ and $\xi_\lambda'=-b_\lambda'/b_\lambda^2$,
\begin{align*}
 \xi_\lambda+\lambda\xi_\lambda'
 &=c(\lambda B_{1,\lambda}+b_\lambda A_{\lambda\lambda})
   -\frac{c\lambda B_{1,\lambda}}{b_\lambda(1+c\lambda B_{1,\lambda})}\\
 &=\frac{cb_\lambda A_{\lambda\lambda}}{1+c\lambda B_{1,\lambda}}
 =\frac{cb_\lambda^2A_{\lambda\lambda}}{1-cb_\lambda^2A_{\lambda\lambda}}.
 \end{align*}
Multiplication by $2/b_\lambda^2$ proves \eqref{eq:bprime}. The numerator is positive because $H((0,\infty))=1$.
\end{proof}

\begin{lemma}[Cross-ridge kernel]\label{lem:kernel}
For positive ridges,
\begin{equation}\label{eq:Kintegral}
 K(\lambda,\mu)=\frac{2cA_{\lambda\mu}}{1-cb_\lambda b_\mu A_{\lambda\mu}}.
\end{equation}
The kernel matrix is positive definite on distinct ridges; consequently $0<\Gamma_{ij}<1$ for $i\ne j$.
\end{lemma}
\begin{proof}
Write $b=b_\lambda$, $e=b_\mu$, $A=A_{\lambda\mu}$ and $L=\int\tau D_\lambda(\tau)D_\mu(\tau)\,dH(\tau)$. First,
\[
 cbeA\le\{cb^2A_{\lambda\lambda}\}^{1/2}
             \{ce^2A_{\mu\mu}\}^{1/2}<1
\]
by Cauchy--Schwarz and Lemma~\ref{lem:fixedpoint}. Subtracting $\xi_\lambda=c\int\tau D_\lambda\,dH$ and $\xi_\mu=c\int\tau D_\mu\,dH$ gives
\[
 \frac{b-e}{be}=c\{(\lambda-\mu)L+(b-e)A\},\qquad
 (b-e)(1-cbeA)=cbe(\lambda-\mu)L.
\]
For $\lambda\ne\mu$,
\begin{align*}
 \frac{\mu\xi_\mu-\lambda\xi_\lambda}{\mu-\lambda}
 &=cA\frac{\mu b-\lambda e}{\mu-\lambda},\\
 \frac{\mu b-\lambda e}{\mu-\lambda}
 &=b-\frac{\lambda cbeL}{1-cbeA}
 =\frac{be}{1-cbeA}.
 \end{align*}
The last equality uses $1-e=ce(\lambda L+bA)$. Substitution in \eqref{eq:K} proves \eqref{eq:Kintegral}; the diagonal case is \eqref{eq:bprime}.

Put $f_j(\tau)=\tau D_{\lambda_j}(\tau)$ and $b_j=b_{\lambda_j}$. The geometric series in the positive denominator gives
\begin{equation*}
 K(\lambda_i,\lambda_j)
 =2\sum_{m=1}^{\infty}c^m(b_i b_j)^{m-1}
       \left\{\int f_i(\tau)f_j(\tau)\,dH(\tau)\right\}^{m}.
\end{equation*}
For arbitrary real $t_1,\ldots,t_J$, its quadratic form equals
\[
 2\sum_{m=1}^{\infty}c^m
 \int\cdots\int\left\{\sum_{j=1}^Jt_jb_j^{m-1}
                   \prod_{a=1}^m f_j(\tau_a)\right\}^{\!2}
 dH(\tau_1)\cdots dH(\tau_m)\ge0.
\]
If this expression is zero, every squared integrand is zero on $(\supp H)^m$. Continuity proves this: a nonzero value would remain nonzero on a neighborhood of positive product measure.

Choose $\tau_0>0$ in $\supp H$ and evaluate at $(\tau_0,\ldots,\tau_0)$. For every $m\ge1$,
\[
 \sum_{j=1}^J\frac{t_j}{b_j}\zeta_j^m=0,
 \qquad \zeta_j=\frac{b_j\tau_0}{\lambda_j+b_j\tau_0}.
\]
The numbers $\zeta_j$ are nonzero and distinct, since
\[
 \frac{d}{d\lambda}\frac{\lambda}{b_\lambda}
 =\frac{b_\lambda-\lambda b_\lambda'}{b_\lambda^2}
 =\frac1{b_\lambda(1+c\lambda B_{1,\lambda})}>0.
\]
The equations for $m=1,\ldots,J$ have coefficient determinant
\[
 \det(\zeta_j^m)_{m,j=1}^J
 =\left(\prod_{j=1}^J\zeta_j\right)\prod_{i<j}(\zeta_j-\zeta_i)\ne0.
\]
Thus all $t_j=0$. Positive definiteness and \eqref{eq:Kintegral} imply
\[
 K_{ii}K_{jj}-K_{ij}^2>0,
 \qquad K_{ij}>0,
 \qquad 0<\frac{K_{ij}}{\sqrt{K_{ii}K_{jj}}}<1.
\]
\end{proof}

\subsection{The local signal measure}\label{app:signal}
Using the eigenvectors $\bm\psi_{r,N}$ from \eqref{eq:alphadependence}, define measures on $[0,C_\Sigma]$ by
\[
 \nu_N=\frac1N\sum_{r=1}^N
    (\bm\psi_{r,N}^{\trans}\mathbf\Omega_N\bm\psi_{r,N})\delta_{\tau_{r,N}},
 \qquad
 \widetilde\nu_T=\sqrt T\sum_{r=1}^N
    (\bm\psi_{r,N}^{\trans}\aalpha_T)^2\delta_{\tau_{r,N}}.
\]
\begin{lemma}[Signal representation]\label{lem:signal}
Assumption~\ref{ass:alt} implies $\nu_N\Rightarrow\nu$ and $\widetilde\nu_T\Rightarrow\nu$ in probability for a unique finite nonnegative measure $\nu$, with
\begin{equation}\label{eq:signalrepresentation}
 m_\Omega(s)=\int\frac{d\nu(\tau)}{s+\tau},\qquad
 \calH_\Omega(\lambda)=\frac1d\int D_\lambda(\tau)\,d\nu(\tau).
\end{equation}
\end{lemma}
\begin{proof}
Take $s_0=1$, and let $M_N=\nu_N([0,C_\Sigma])$ and $\widetilde M_T=\widetilde\nu_T([0,C_\Sigma])$. The two parts of \eqref{eq:originalsignal} imply
\begin{align*}
 M_N&\le(C_\Sigma+s_0)\int\frac{d\nu_N(\tau)}{s_0+\tau}\le C,\\
 \widetilde M_T&\le(C_\Sigma+s_0)\int\frac{d\widetilde\nu_T(\tau)}{s_0+\tau}=O_p(1).
 \end{align*}
Bounded masses on a compact interval give weakly convergent subsequences of finite measures; this is the compactness criterion in \citet[Sections 1 and 5]{Billingsley1999}.

Every weak cluster limit $\nu^{(a)}$ satisfies, by bounded continuity,
\[
 \int\frac{d\nu^{(a)}(\tau)}{s+\tau}
 =\lim_{\ell\to\infty}\int\frac{d\nu_{N_\ell}(\tau)}{s+\tau}=m_\Omega(s).
\]
For two cluster limits and $s>C_\Sigma$, subtraction and a geometric expansion give
\[
 0=\sum_{r=0}^{\infty}\frac{(-1)^r}{s^{r+1}}
       \left\{\int\tau^r\,d\nu^{(1)}(\tau)-\int\tau^r\,d\nu^{(2)}(\tau)\right\}.
\]
The series converges absolutely and uniformly in $\tau$. Multiplication by $s$, followed by $s\to\infty$, identifies the total masses.

Suppose the moments below order $m$ agree. Multiplication by $s^{m+1}$ gives the difference of the $m$th moments plus a remainder bounded by
\[
 2M\sum_{r=m+1}^{\infty}\frac{C_\Sigma^r}{s^{r-m}}
 =\frac{2M C_\Sigma^{m+1}}{s(1-C_\Sigma/s)}\longrightarrow0,
\]
where $M$ bounds both masses. Induction identifies every moment.

To pass from moments to continuous functions, let $f$ be continuous and $\omega_f(a)=\sup_{|x-y|\le a}|f(x)-f(y)|$. Its Bernstein polynomial satisfies
\[
 B_mf(\tau)=\sum_{r=0}^mf(C_\Sigma r/m)\binom mr
    (\tau/C_\Sigma)^r(1-\tau/C_\Sigma)^{m-r}.
\]
For $V_{m,\tau}\sim\operatorname{Bin}(m,\tau/C_\Sigma)$ and $a>0$,
\begin{align*}
 |B_mf(\tau)-f(\tau)|
 &\le\omega_f(C_\Sigma a)+2\|f\|_\infty\PP(|V_{m,\tau}/m-\tau/C_\Sigma|>a)\\
 &\le\omega_f(C_\Sigma a)+\frac{\|f\|_\infty}{2ma^2}.
 \end{align*}
Taking $a=m^{-1/4}$ proves uniform approximation. Since all polynomial integrals agree, all continuous-function integrals agree, identifying a single measure $\nu$.

Every deterministic subsequence therefore has the same cluster limit. For the random measures, enumerate the positive rationals as $s_1,s_2,\ldots$.

From an arbitrary subsequence, choose $T_m$ so that
\[
 \PP\left\{\max_{\ell\le m}\left|
 \int\frac{d\widetilde\nu_{T_m}(\tau)}{s_\ell+\tau}-m_\Omega(s_\ell)
 \right|>m^{-1}\right\}\le2^{-m}.
\]
If $E_m$ denotes the event in this display, then
\[
 \PP\Big(\bigcup_{m\ge M}E_m\Big)\le\sum_{m\ge M}2^{-m}=2^{1-M}\longrightarrow0.
\]
All rational-argument transforms therefore converge simultaneously almost surely. Since the enumeration includes $s_0=1$, the mass bound gives almost-sure bounded masses on this subsequence.

For $s,t\ge s_*>0$,
\[
 \left|\int\frac{d\widetilde\nu_{T_m}(\tau)}{s+\tau}
            -\int\frac{d\widetilde\nu_{T_m}(\tau)}{t+\tau}\right|
 \le s_*^{-2}|s-t|\widetilde M_{T_m}.
\]
Thus convergence extends from the positive rationals to every positive argument. Compactness and the preceding uniqueness proof give almost-sure weak convergence along this further subsequence.

If convergence in probability failed for a continuous $f$, there would be $\epsilon,c_0>0$ and a subsequence with
\[
 \PP\left(\left|\int f(\tau)\,d\widetilde\nu_T(\tau)
                -\int f(\tau)\,d\nu(\tau)\right|>\epsilon\right)>c_0.
\]
Its almost-surely convergent further subsequence contradicts this inequality. This proves the random-measure assertion.

The following explicit finite-net bound will be used for complex arguments. If $h_z(\tau)$ is continuous and $|h_z(\tau)-h_{z'}(\tau)|\le L|z-z'|$, then
\begin{align}
 &\sup_{z\in\mathcal S}\left|\int h_z(\tau)\,d\widetilde\nu_T(\tau)
                         -\int h_z(\tau)\,d\nu(\tau)\right|\notag\\
 &\quad\le\max_{\ell\le \mathfrak n_\epsilon}\left|\int h_{z_\ell}(\tau)\,d\widetilde\nu_T(\tau)
                         -\int h_{z_\ell}(\tau)\,d\nu(\tau)\right|
       +L\epsilon(\widetilde M_T+\nu([0,C_\Sigma])),\label{eq:signalnet}
\end{align}
where $\{z_\ell\}_{\ell=1}^{\mathfrak n_\epsilon}$ is an $\epsilon$-net of the compact parameter set $\mathcal S$. First increasing $T$ and then decreasing $\epsilon$ proves uniform convergence.

Finally, diagonalization and substitution of $s=\lambda/b_\lambda$ give
\[
 m_\Omega(s)=\int\frac{d\nu(\tau)}{s+\tau},\qquad
 \frac{m_\Omega(\lambda/b_\lambda)}{db_\lambda}
 =\frac1d\int\frac{d\nu(\tau)}{\lambda+b_\lambda\tau},
\]
which proves \eqref{eq:signalrepresentation}.
\end{proof}
\section{Proofs of the joint ridge limits}\label{app:ridgeproofs}
\subsection{General spectral input and regression representation}\label{app:input}
Write $n=T-p-1$, $r_T=T/n$, $\bm q_T=\bm v_T/\|\bm v_T\|$ and $\overline{\mathbf S}_T=r_T\SSS$. For a real analytic spectral function $f$, define
\[
 M_{n,T}(f)=\frac1n(\mathbf Y\bm q_T)^{\trans}
                     f(\overline{\mathbf S}_T)(\mathbf Y\bm q_T).
\]
Lemma~\ref{lem:projection} gives the exact conversion to the original RHT:
\begin{equation}\label{eq:exactmapping}
 \mathbf Y\bm q_T=\sqrt{T/d_T}\,\ahat,\qquad
 M_{n,T}\{(x+r_T\lambda)^{-1}\}
 =\frac{T}{nd_T}\ahat^{\trans}\{r_T(\SSS+\lambda\I_N)\}^{-1}\ahat
 =W_T(\lambda).
\end{equation}
The general linear contrast is the intercept, its rank is one, and the design rank is $p+1$. The dimension and design conversions are
\[
 \frac{N}{n}-c=\left(\frac{N}{T}-c\right)+\frac{N(p+1)}{nT},\qquad
 \sqrt n\left|\frac{N}{n}-c\right|
 \le\sqrt T|c_T-c|+CT^{-1/2}\to0,
\]
\[
 n^{-1}\mathbf X^{\trans}\mathbf X=r_T\mathbf D_T\to\mathbf D\succ0,
 \qquad \bm e_1^{\trans}(n^{-1}\mathbf X^{\trans}\mathbf X)^{-1}\bm e_1\to d>0.
\]
Together with $\sqrt nW_1(H_N,H)\to0$, these formulas verify C1--C6 of \citet{LAP2020} from Assumption~\ref{ass:base} and \eqref{eq:designscope}.

Let $F_\infty$ be the limiting sample spectral distribution and define, off its nonnegative support,
\[
 m_\infty(z)=\int\frac{dF_\infty(x)}{x-z},\qquad
 b(z)=1-c-czm_\infty(z),\qquad \Theta(z)=b(z)^{-1}.
\]
The sample spectral law and its transform equation are given in Section 2.1 of \citet{LAP2020}. In particular, $b(-\lambda)=b_\lambda$ and $\Theta(-\lambda)=1+\xi_\lambda$.

Set
\begin{align}
 \mathcal D(z_1,z_2)&=\Theta(z_1)\Theta(z_2)
   \left\{\frac{z_1\Theta(z_1)-z_2\Theta(z_2)}{z_1-z_2}-1\right\},\notag\\
 \mathcal D(z,z)&=\Theta(z)^2\{\Theta(z)+z\Theta'(z)-1\},\qquad
 \mathscr S_\nu(z)=\int\frac{d\nu(\tau)}{b(z)\tau-z}.\label{eq:contourkernel}
\end{align}
The diagonal definition removes the divided-difference singularity. For a positively oriented contour $\mathcal C$ around the limiting spectrum, define
\begin{align}
 \mathfrak A(f)&=-\frac1{2\pi\ii}\oint_{\mathcal C}f(z)\{\Theta(z)-1\}\,dz,\notag\\
 \mathfrak V(f_1,f_2)&=\frac2{(2\pi\ii)^2}
       \oint_{\mathcal C}\oint_{\mathcal C}
       f_1(z_1)f_2(z_2)\mathcal D(z_1,z_2)\,dz_1dz_2,\notag\\
 \mathfrak H_\nu(f)&=-\frac1{2\pi\ii d}\oint_{\mathcal C}f(z)\mathscr S_\nu(z)\,dz.\label{eq:spectralfunctionals}
\end{align}
The contour encloses $[0,C_\Sigma(1+\sqrt c)^2]$ and excludes the negative poles of the chosen ridge functions.

The companion transform $\underline m(z)=-(1-c)/z+cm_\infty(z)$ is the Stieltjes transform of a probability distribution on a bounded nonnegative interval. Its imaginary part has the same sign as $\operatorname{Im}z$, and
\[
 b(z)=-z\underline m(z),\qquad
 b(z)\tau-z=-z\{1+\tau\underline m(z)\}.
\]

Thus $b(z)\tau-z\ne0$ for nonreal $z$ and $\tau\ge0$. For real $z<0$, $b(z)>0$, so the denominator is positive.

For large $|z|$, the bounded support gives $b(z)=1+O(|z|^{-1})$. Choose the right edge of a rectangular contour sufficiently far to the right. Then
\[
 \min_{z\in\mathcal C,\,0\le\tau\le C_\Sigma}|b(z)\tau-z|=:a_\mathcal C>0,
 \qquad
 \sup_{z\in\mathcal C,\tau}
 \left|\frac{\partial}{\partial z}\frac1{b(z)\tau-z}\right|<\infty.
\]
For any deterministic alpha sequence with $\widetilde\nu_T\Rightarrow\nu$, the signal condition (2.3) of \citet{LAP2020} is therefore verified by
\begin{align}
 \sqrt n\,\aalpha_T^{\trans}\{b(z)\Sig-z\I_N\}^{-1}\aalpha_T
 &=\sqrt{n/T}\int\frac{d\widetilde\nu_T(\tau)}{b(z)\tau-z},\notag\\
 \sup_{z\in\mathcal C}|\sqrt n\,\aalpha_T^{\trans}\{b(z)\Sig-z\I_N\}^{-1}\aalpha_T
             -\mathscr S_\nu(z)|
 &\le C\widetilde M_T/T+r_{\nu,T}(\mathcal C),\label{eq:deterministicaltinput}
\end{align}
where $r_{\nu,T}(\mathcal C)$ is the supremum of the difference between the two integrals with measures $\widetilde\nu_T$ and $\nu$. It tends to zero by \eqref{eq:signalnet}.

\begin{lemma}[General linear spectral limit; \citealp{LAP2020}]\label{lem:external}
Under the conditions just verified, for real analytic $f$ with $\mathfrak V(f,f)>0$,
\begin{equation}\label{eq:externalCLT}
 \sqrt n\{M_{n,T}(f)-\mathfrak A(f)\}
 \dto \mathcal{N}\{\mathfrak H_\nu(f),\mathfrak V(f,f)\}.
\end{equation}
Under $H_0$, take $\nu=0$.
\end{lemma}
Lemma~\ref{lem:external} is the rank-one specialization of Theorems 2.1 and 2.3 of \citet{LAP2020} under the regression representation above.

For the marginal corrections, put $c_n=N/n$ and
\[
 \overline m_T(z)=N^{-1}\tr(\overline{\mathbf S}_T-z\I_N)^{-1},\qquad
 \widehat\Theta_n(z)=\{1-c_n-c_nz\overline m_T(z)\}^{-1}.
\]
Define, for $\lambda>0$,
\[
 \overline\xi_T(\lambda)=\widehat\Theta_n(-\lambda)-1,\qquad
 \overline\Delta_T(\lambda)
 =2\{\overline\xi_T(\lambda)+\lambda\overline\xi_T'(\lambda)\}
       \{1+\overline\xi_T(\lambda)\}^2.
\]
Lemma 2.2 of \citet{LAP2020}, applied to the marginal centering and variance, gives
\begin{equation}\label{eq:externalestimation}
 \sqrt n\,|\overline\xi_T(\lambda)-\xi_\lambda|\pto0,
 \qquad
 \sqrt n\,|\overline\Delta_T(\lambda)-\Delta_\lambda|\pto0.
\end{equation}

\subsection{The cross-ridge covariance functional}
For $f_\lambda(x)=(x+\lambda)^{-1}$, the single-ridge identities in \citet{CKMS2025} and \citet[Lemma 2.1 and Theorem 2.3]{LAP2020} give
\begin{equation}\label{eq:marginalfunctionals}
 \mathfrak A(f_\lambda)=\xi_\lambda,\qquad
 \mathfrak H_\nu(f_\lambda)=\calH_\Omega(\lambda).
\end{equation}
The covariance between two ridge functionals is obtained as follows.
\begin{lemma}[Cross-ridge covariance]\label{lem:residues}
For $\lambda,\mu>0$,
\begin{equation}\label{eq:residues}
 \mathfrak V(f_\lambda,f_\mu)=K(\lambda,\mu).
\end{equation}
\end{lemma}
\begin{proof}
Let $L(z)$ be analytic outside $\mathcal C$ with $L(z)=O(|z|^{-1})$. Enlarge $\mathcal C$ to a circle $\mathcal C_R$ that also encloses $-\lambda$. The residue at $-\lambda$ is $L(-\lambda)$, and
\[
 \left|\oint_{\mathcal C_R}\frac{L(z)}{z+\lambda}\,dz\right|
 \le2\pi R\sup_{|z|=R}\frac{|L(z)|}{|z+\lambda|}=O(R^{-1}).
\]
Subtracting the two contour integrals and letting $R\to\infty$ yields the residue identity used in \citet[Lemma 2.1]{LAP2020}:
\begin{equation}\label{eq:residueidentity}
 -\frac1{2\pi\ii}\oint_{\mathcal C}\frac{L(z)}{z+\lambda}\,dz=L(-\lambda).
\end{equation}
For fixed $z_2$, $\mathcal D(z_1,z_2)=O(|z_1|^{-1})$ as $|z_1|\to\infty$; the same holds with the arguments interchanged. The divided difference has a removable singularity at $z_1=z_2$. Applying \eqref{eq:residueidentity} in each argument gives
\[
 \mathfrak V(f_\lambda,f_\mu)=2\mathcal D(-\lambda,-\mu).
\]
For $\lambda\ne\mu$, substitution yields
\begin{align*}
 2\mathcal D(-\lambda,-\mu)
 &=2(1+\xi_\lambda)(1+\xi_\mu)
   \left\{\frac{-\lambda(1+\xi_\lambda)+\mu(1+\xi_\mu)}{\mu-\lambda}-1\right\}\\
 &=2(1+\xi_\lambda)(1+\xi_\mu)
       \frac{\mu\xi_\mu-\lambda\xi_\lambda}{\mu-\lambda}
 =K(\lambda,\mu).
 \end{align*}
As $\mu\to\lambda$, the divided difference tends to $\xi_\lambda+\lambda\xi_\lambda'$, so that $\mathfrak V(f_\lambda,f_\lambda)=\Delta_\lambda$.
\end{proof}

\begin{lemma}[Removing the random-alpha conditioning]\label{lem:randomCLT}
The scalar limit \eqref{eq:externalCLT} remains valid for the independent random alpha draw in Assumption~\ref{ass:alt}, with the deterministic measure $\nu$ from Lemma~\ref{lem:signal}.
\end{lemma}
\begin{proof}
For real $t$, write $A_T=\sqrt n\{M_{n,T}(f)-\mathfrak A(f)\}$ and
\[
 \chi_T(t\mid\aalpha_T)=\E(e^{\ii tA_T}\mid\aalpha_T),\qquad
 \chi(t)=\exp\{\ii t\mathfrak H_\nu(f)-t^2\mathfrak V(f,f)/2\}.
\]
From every subsequence, Lemma~\ref{lem:signal} provides a further subsequence on which $\widetilde\nu_T\Rightarrow\nu$ almost surely.

For almost every realized alpha sequence, \eqref{eq:deterministicaltinput} holds uniformly on $\mathcal C$. Independence preserves the innovation law under conditioning, so Lemma~\ref{lem:external} gives
\[
 \chi_T(t\mid\aalpha_T)\longrightarrow\chi(t)\quad\text{almost surely},
 \qquad |\chi_T(t\mid\aalpha_T)-\chi(t)|\le2.
\]
Bounded convergence implies
\[
 |\E e^{\ii tA_T}-\chi(t)|
 \le\E|\chi_T(t\mid\aalpha_T)-\chi(t)|\longrightarrow0
\]
along that further subsequence. Every subsequence has the same limit, proving unconditional characteristic-function convergence and hence \eqref{eq:externalCLT}.
\end{proof}

\subsection{Marginal normalization bounds}
For the fixed ridges, write
\[
 \begin{aligned}
 U_{j,T}&=\sqrt T\{W_T(\lambda_j)-\xi_{\lambda_j}\},&
 \overline U_{j,T}&=\sqrt n\{M_{n,T}(f_{\lambda_j})-\xi_{\lambda_j}\},\\
 \widehat\xi_j&=\widehat\xi_T(\lambda_j),&
 \overline\xi_{j,T}&=\overline\xi_T(\lambda_j),\\
 \widehat\Delta_j&=\widehat\Delta_T(\lambda_j),&
 \overline\Delta_{j,T}&=\overline\Delta_T(\lambda_j).
 \end{aligned}
\]
Define the scalar correction errors by
\[
 \epsilon_{\xi,T}=\sqrt n\max_j|\overline\xi_{j,T}-\xi_{\lambda_j}|,
 \qquad
 \epsilon_{\Delta,T}=\sqrt n\max_j|\overline\Delta_{j,T}-\Delta_{\lambda_j}|.
\]
Both errors converge to zero in probability by \eqref{eq:externalestimation}. Let $b_* =\min_jb_{\lambda_j}>0$ and $\Delta_* =\min_j\Delta_{\lambda_j}>0$.

\begin{lemma}[Marginal normalization]\label{lem:normalization}
Under either hypothesis,
\begin{equation}\label{eq:normalizationerror}
 \max_j|U_{j,T}-\overline U_{j,T}|=O_p(T^{-1/2}),\qquad
 \max_j\{|\widehat\xi_j-\overline\xi_{j,T}|
          +|\widehat\Delta_j-\overline\Delta_{j,T}|\}\le C/T
\end{equation}
on events whose probabilities tend to one.
\end{lemma}
\begin{proof}
The fixed-function version of \eqref{eq:exactmapping} gives
\[
 M_{n,T}(f_\lambda)=W_T(\lambda/r_T),\qquad
 |r_T-1|=\frac{p+1}{n}\le C/T.
\]
On a fixed compact positive ridge interval, differentiation yields
\[
 \QQ'(\lambda)=-\QQ(\lambda)^2,\qquad
 |W_T'(\lambda)|\le\frac{\|\ahat\|^2}{d_T\lambda^2},\qquad
 |W_T(\lambda)-W_T(\lambda/r_T)|\le CT^{-1}\|\ahat\|^2.
\]
Since $|\sqrt T-\sqrt n|=(p+1)/(\sqrt T+\sqrt n)$, Lemma~\ref{lem:scale} implies
\[
 |U_{j,T}-\overline U_{j,T}|
 \le CT^{-1/2}(1+\|\ahat\|^2)=O_p(T^{-1/2}).
\]
Let $s_1,\ldots,s_n$ be the positive eigenvalues of $\SSS$, padded with zeros to length $n$. Define the scalar spectral moments
\[
 h_T(\lambda)=\frac1T\sum_{\ell=1}^n\frac{s_\ell}{s_\ell+\lambda},\qquad
 v_T(\lambda)=\frac1T\sum_{\ell=1}^n
                        \left(\frac{s_\ell}{s_\ell+\lambda}\right)^2.
\]
The identities $h_T=c_T-c_T\lambda\widehat m_T$ and $h_T+\lambda h_T'=v_T$ imply
\begin{equation}\label{eq:marginalmoments}
 \widehat\xi_T(\lambda)=\frac{h_T(\lambda)}{1-h_T(\lambda)},\qquad
 \widehat\Delta_T(\lambda)
 =\frac{2\{v_T(\lambda)-h_T(\lambda)^2\}}{\{1-h_T(\lambda)\}^4}.
\end{equation}
Indeed, differentiating $h_T/(1-h_T)$ gives
\[
 \widehat\xi_T+\lambda\widehat\xi_T'
 =\frac{h_T(1-h_T)+\lambda h_T'}{(1-h_T)^2}
 =\frac{v_T-h_T^2}{(1-h_T)^2}.
\]
The $n$-normalized moments of the eigenvalues $r_Ts_\ell$ satisfy
\[
 \overline h_T(\lambda)=r_T h_T(\lambda/r_T),\qquad
 \overline v_T(\lambda)=r_T v_T(\lambda/r_T).
\]
For $\lambda\ge\ell_0>0$,
\[
 0\le h_T,v_T\le1,\qquad
 |h_T'(\lambda)|\le\frac1{4\ell_0},\qquad
 |v_T'(\lambda)|\le\frac1{2\ell_0},\qquad
 |\overline h_T-h_T|+|\overline v_T-v_T|\le C/T.
\]
On events where $1-h_T$ and $1-\overline h_T$ exceed $b_*/2$, the two rational maps in \eqref{eq:marginalmoments} have bounded derivatives. Their mean-value bounds prove \eqref{eq:normalizationerror}.

In particular, \eqref{eq:externalestimation} gives
\begin{align}
 \sqrt T\max_j|\widehat\xi_j-\xi_{\lambda_j}|
 &\le CT^{-1/2}+\sqrt{T/n}\,\epsilon_{\xi,T},\notag\\
 r_{\Delta,T}:=\max_j|\widehat\Delta_j-\Delta_{\lambda_j}|
 &\le C/T+n^{-1/2}\epsilon_{\Delta,T}=O_p(T^{-1/2}).\label{eq:feasibleerrors}
\end{align}
For $a,b\ge\Delta_*/2$, $|a^{-1/2}-b^{-1/2}|\le C|a-b|$. With $\bm Z_T=(Z_T(\lambda_1),\ldots,Z_T(\lambda_J))^{\trans}$, $\bm U_T=(U_{1,T},\ldots,U_{J,T})^{\trans}$ and $\mathbf D_\Delta=\diag(\sqrt{\Delta_{\lambda_1}},\ldots,\sqrt{\Delta_{\lambda_J}})$,
\[
 Z_T(\lambda_j)-\frac{U_{j,T}}{\sqrt{\Delta_{\lambda_j}}}
 =-\frac{\sqrt T(\widehat\xi_j-\xi_{\lambda_j})}{\sqrt{\widehat\Delta_j}}
   +U_{j,T}\{\widehat\Delta_j^{-1/2}-\Delta_{\lambda_j}^{-1/2}\}.
\]
Hence
\begin{equation}\label{eq:standardizationerror}
 \|\bm Z_T-\mathbf D_\Delta^{-1}\bm U_T\|_\infty
 \le C\{T^{-1/2}+\epsilon_{\xi,T}
 +(T^{-1}+n^{-1/2}\epsilon_{\Delta,T})\|\bm U_T\|_\infty\}.
\end{equation}
Only the marginal corrections enter these bounds. Their distributions are unchanged by the local alpha because $\SSS=T^{-1}\mathbf E\mathbf M_X\mathbf E^{\trans}$.
\end{proof}

\subsection{Joint null and local-alternative limits}
\begin{proof}[Proof of Theorem~\ref{thm:null}]
Fix $\bm t\in\R^J$ and set
\[
 f_{\bm t}(x)=\sum_{j=1}^J\frac{t_j}{\sqrt{\Delta_{\lambda_j}}}\frac1{x+\lambda_j}.
\]
Every pole is strictly negative, so the function is analytic around the nonnegative spectrum. Linearity of $M_{n,T}$ and \eqref{eq:marginalfunctionals}--\eqref{eq:residues} give
\begin{align*}
 \sqrt n\{M_{n,T}(f_{\bm t})-\mathfrak A(f_{\bm t})\}
 &=\sum_j\frac{t_j\overline U_{j,T}}{\sqrt{\Delta_{\lambda_j}}},\\
 \mathfrak V(f_{\bm t},f_{\bm t})
 &=\sum_{i,j}\frac{t_it_jK(\lambda_i,\lambda_j)}
                      {\sqrt{\Delta_{\lambda_i}\Delta_{\lambda_j}}}
 =\bm t^{\trans}\Gam\bm t.
 \end{align*}
For $\bm t\ne0$, this variance is positive by Lemma~\ref{lem:kernel}. Lemma~\ref{lem:external} therefore implies
\[
 \E\exp\left(\ii\sum_j\frac{t_j\overline U_{j,T}}{\sqrt{\Delta_{\lambda_j}}}\right)
 \longrightarrow\exp(-\bm t^{\trans}\Gam\bm t/2).
\]
The scalar limits also imply tightness of the finite vector: for $M>0$,
\[
 \PP(\max_j|\overline U_{j,T}|>M)\le\sum_j\PP(|\overline U_{j,T}|>M),
 \qquad
 \lim_{M\to\infty}\limsup_T\PP(\max_j|\overline U_{j,T}|>M)=0.
\]
Equations~\eqref{eq:normalizationerror} and \eqref{eq:standardizationerror} show that
\[
 \left|\bm t^{\trans}\bm Z_T-
       \sum_j\frac{t_j\overline U_{j,T}}{\sqrt{\Delta_{\lambda_j}}}\right|
 \le C_{\bm t}\epsilon_{\xi,T}+O_p(T^{-1/2})\pto0.
\]
For any $a>0$, the difference between the corresponding characteristic functions is at most $a+2\PP(|\text{difference}|>a)$.

Letting $T\to\infty$ and then $a\downarrow0$ gives the characteristic function of $\mathcal{N}_J(\bm0,\Gam)$. The multivariate continuity theorem proves \eqref{eq:jointnull}; see \citet{Billingsley1999}.

Lemma~\ref{lem:kernel} establishes the stated properties of the population correlation matrix.
\end{proof}

\begin{proof}[Proof of Theorem~\ref{thm:alternative}]
For the same $f_{\bm t}$, \eqref{eq:marginalfunctionals} and \eqref{eq:signalrepresentation} give
\[
 \mathfrak H_\nu(f_{\bm t})
 =\sum_j\frac{t_j}{\sqrt{\Delta_{\lambda_j}}}
         \frac1d\int\frac{d\nu(\tau)}{\lambda_j+b_{\lambda_j}\tau}
 =\bm t^{\trans}\bm\delta_\Omega.
\]
Lemma~\ref{lem:randomCLT} removes the alpha conditioning. Thus
\[
 \E\exp\left(\ii\sum_j\frac{t_j\overline U_{j,T}}{\sqrt{\Delta_{\lambda_j}}}\right)
 \longrightarrow\exp\{\ii\bm t^{\trans}\bm\delta_\Omega
                           -\bm t^{\trans}\Gam\bm t/2\}.
\]
The shifted scalar normal limits make the finite vector tight. The normalization and studentization bounds give
\[
 \left|\bm t^{\trans}\bm Z_T-
       \sum_j\frac{t_j\overline U_{j,T}}{\sqrt{\Delta_{\lambda_j}}}\right|
 \le C_{\bm t}\epsilon_{\xi,T}+O_p(T^{-1/2})\pto0.
\]
The characteristic-function inequality used in the preceding proof yields \eqref{eq:jointalt}.

The covariance remains $\Gam$ because the local signal enters only $\mathfrak H_\nu(f_{\bm t})$.
\end{proof}

\subsection{The estimated ridge scale}
\begin{proof}[Proof of Corollary~\ref{cor:scale}]
Since $\sum_{t=1}^T(\widehat{\bm a}_t-\ahat)=\bm0$, the covariance definition gives
\begin{align*}
 \frac1{NT}\sum_{t=1}^T\|\widehat{\bm a}_t\|^2
 &=\frac1{NT}\sum_{t=1}^T\|\widehat{\bm a}_t-\ahat\|^2
   +\frac{\|\ahat\|^2}{N}
 =\widehat s_T^{\mathrm{raw}},\\
 \widehat s_T^{\mathrm{raw}}-\mu_1
 &=R_{s,T}+r_{s,T}+\frac{\|\ahat\|^2}{N}.
 \end{align*}
Under either hypothesis, Lemma~\ref{lem:scale} yields $\E\|\ahat\|^2\le C$. Hence, for every $x>0$,
\[
 \PP\left\{\frac{\|\ahat\|^2}{N}>\frac{x}{T}\right\}
 \le\frac{T\E\|\ahat\|^2}{Nx}\le\frac Cx,
 \qquad \frac{\|\ahat\|^2}{N}=O_p(T^{-1}).
\]
Choose $[\ell,u]\subset(0,\infty)$ with every $\lambda_j=a_j\mu_1$ in its interior. The bounds on $R_{s,T}$ and $r_{s,T}$ in Lemma~\ref{lem:scale} imply
\begin{align}
 \sqrt T\max_j|\lambda_{j,T}-\lambda_j|
 &\le C\sqrt T|R_{s,T}|+CT^{-1/2}+C\sqrt TW_1(H_N,H)
       +\frac{C\sqrt T}{N}\|\ahat\|^2\notag\\
 &=O_p(T^{-1/2})+O\{\sqrt TW_1(H_N,H)\}.\label{eq:randomridgerate}
\end{align}
All random ridges belong to $[\ell,u]$ with probability tending to one. On this interval the spectral weights satisfy, for integers $r\ge0$,
\[
 \left|\frac{d^r}{d\lambda^r}\frac{s}{s+\lambda}\right|
 =\frac{r!s}{(s+\lambda)^{r+1}}\le r!\ell^{-r},\qquad s\ge0.
\]
Moreover, $\widehat b_T(\lambda)=1-h_T(\lambda)$ is increasing. Pointwise convergence at $\ell$ gives
\[
 \inf_{\lambda\in[\ell,u]}\widehat b_T(\lambda)
 \ge\widehat b_T(\ell)\ge b_\ell/2
\]
with probability tending to one. The marginal rational expressions in \eqref{eq:marginalmoments} consequently have uniformly bounded first derivatives.

A finite mesh with spacing at most $a$ gives
\[
 \sup_{\lambda\in[\ell,u]}|\widehat\Delta_T(\lambda)-\Delta_\lambda|
 \le\max_{\lambda\text{ on mesh}}|\widehat\Delta_T(\lambda)-\Delta_\lambda|+Ca.
\]
Use the pointwise estimates \eqref{eq:externalestimation}, then let $a\downarrow0$. Since $\inf_{[\ell,u]}\Delta_\lambda>0$, all empirical variances are uniformly positive on an event of probability tending to one.

Differentiation of the studentized statistic gives
\begin{align*}
 Z_T'(\lambda)
 &=\sqrt T\left\{\frac{W_T'(\lambda)-\widehat\xi_T'(\lambda)}
                        {\sqrt{\widehat\Delta_T(\lambda)}}
 -\frac{(W_T(\lambda)-\widehat\xi_T(\lambda))\widehat\Delta_T'(\lambda)}
                       {2\widehat\Delta_T(\lambda)^{3/2}}\right\},\\
 \sup_{[\ell,u]}|Z_T'(\lambda)|&\le C\sqrt T(1+\|\ahat\|^2).
 \end{align*}
Here $|W_T|\le\ell^{-1}\|\ahat\|^2$ and $|W_T'|\le\ell^{-2}\|\ahat\|^2$. The mean-value theorem and \eqref{eq:randomridgerate} imply
\begin{align*}
 \max_j\left|Z_T(a_j\widehat s_T^{\mathrm{raw}})-Z_T(a_j\mu_1)\right|
 &\le \max_j|\lambda_{j,T}-\lambda_j|
       \sup_{\lambda\in[\ell,u]}|Z_T'(\lambda)|\\
 &=O_p\{T^{-1/2}+\sqrt T W_1(H_N,H)\}.
 \end{align*}
Both terms tend to zero under Assumption~\ref{ass:base}. The characteristic-function perturbation bound used in the preceding joint-limit proofs therefore gives the same null and local-alternative limits for $\lambda_{j,T}=a_j\widehat s_T^{\mathrm{raw}}$.
\end{proof}
\section{Cauchy transformation and size proofs}\label{app:Cauchy}
\subsection{Elementary Gaussian bounds}
\begin{lemma}[Normal tails]\label{lem:normal}
For $x>0$,
\begin{equation}\label{eq:mills}
 \frac{x\phi(x)}{1+x^2}\le1-\Phi(x)\le\frac{\phi(x)}x.
\end{equation}
If $V\sim \mathcal{N}(0,\sigma^2)$, then $\PP(V>y)\le e^{-y^2/(2\sigma^2)}$ for $y>0$.
\end{lemma}
\begin{proof}
Since $\phi'(u)=-u\phi(u)$, integration by parts gives
\begin{align*}
 1-\Phi(x)&\le\frac1x\int_x^\infty u\phi(u)\,du=\frac{\phi(x)}x,\\
 1-\Phi(x)&=\frac{\phi(x)}x-\int_x^\infty\frac{\phi(u)}{u^2}\,du
 \ge\frac{\phi(x)}x-\frac{1-\Phi(x)}{x^2}.
 \end{align*}
Rearranging the second inequality proves the lower bound. Completing the square gives $\E e^{tV}=e^{t^2\sigma^2/2}$, so
\[
 \PP(V>y)\le\inf_{t>0}\exp(-ty+t^2\sigma^2/2)
 =\exp\{-y^2/(2\sigma^2)\}.\qedhere
\]
\end{proof}

\subsection{Density of the limiting Cauchy transform}
Let $m_g=\min\{\pi\phi(1),[\pi\phi(1)]^{-1}\}>0$ and $L_J=J/(m_g\sqrt{2\pi})$.
\begin{lemma}[A bounded positive density]\label{lem:continuity}
If $\mathbf R$ is a correlation matrix with nonnegative entries, then $\mathscr C(\bm\delta,\mathbf R)$ has a continuous density $f_{\bm\delta,\mathbf R}$ satisfying
\begin{equation}\label{eq:densitybound}
 0<f_{\bm\delta,\mathbf R}(x)\le L_J\qquad(x\in\R).
\end{equation}
\end{lemma}
\begin{proof}
The derivative of $g$ and Lemma~\ref{lem:normal} imply
\[
 g'(x)=\pi\phi(x)\{1+g(x)^2\}>0,\qquad
 g'(x)\ge\pi\phi(1)\quad(|x|\le1).
\]
For $x\ge1$, $\sin u\le u$ and \eqref{eq:mills} give
\[
 g'(x)=\frac{\pi\phi(x)}{\sin^2\{\pi(1-\Phi(x))\}}
 \ge\frac{\phi(x)}{\pi\{1-\Phi(x)\}^2}
 \ge\frac{x^2}{\pi\phi(x)}\ge[\pi\phi(1)]^{-1}.
\]
The derivative is even, so $g'(x)\ge m_g$ on $\R$.

For $\bm G\sim \mathcal{N}_J(\bm0,\mathbf R)$, define
\[
 V=\one^{\trans}\bm G,\quad
 \sigma_V^2=\one^{\trans}\mathbf R\one\in[J,J^2],\quad
 a_j=\frac{(\mathbf R\one)_j}{\sigma_V^2}\ge\sigma_V^{-2},\quad
 R_j^\circ=G_j-a_jV.
\]
Then $\Cov(V,\bm R^\circ)=\bm0$. The Gaussian characteristic function factorizes:
\[
 \E e^{\ii tV+\ii\bm u^{\trans}\bm R^\circ}
 =e^{-t^2\sigma_V^2/2}
   e^{-\bm u^{\trans}\Cov(\bm R^\circ)\bm u/2}.
\]
Thus $V$ and $\bm R^\circ$ are independent, including when $\mathbf R$ is singular.

For fixed $\bm r$, define $h_{\bm r}(v)=\sum_jw_jg(a_jv+r_j+\delta_j)$. Then
\[
 h_{\bm r}'(v)\ge m_g\sum_jw_ja_j\ge m_g/\sigma_V^2>0,
 \qquad
 \lim_{v\to\pm\infty}h_{\bm r}(v)=\pm\infty.
\]
The inverse $v_x=h_{\bm r}^{-1}(x)$ is differentiable. A one-dimensional change of variable gives the conditional density
\[
 f(x\mid\bm r)=\frac{\phi(v_x/\sigma_V)}
                         {\sigma_V h_{\bm r}'(v_x)},\qquad
 0<f(x\mid\bm r)\le\frac{\sigma_V}{m_g\sqrt{2\pi}}\le L_J.
\]
Averaging this density over $\bm R^\circ$ proves \eqref{eq:densitybound}; dominated convergence proves its continuity. In particular,
\begin{equation}\label{eq:intervalprob}
 \PP\{x<\mathscr C(\bm\delta,\mathbf R)\le x+h\}\le L_Jh
 \qquad(x\in\R,
 h>0).
\end{equation}
Its strictly positive density gives full support and a strictly increasing distribution function.
\end{proof}

\subsection{Passing the joint Gaussian limit through the transformation}
For $\bm Y_\Omega=\bm G+\bm\delta_\Omega$, $\bm G\sim \mathcal{N}_J(\bm0,\Gam)$, define the bounded-Lipschitz approximation error
\begin{equation}\label{eq:BLerror}
 \mathfrak d_T(\Omega)=\sup_h
 \big|\E h(\bm Z_T)-\E h(\bm Y_\Omega)\big|,
 \quad |h|\le1,\quad |h(\bm x)-h(\bm y)|\le\|\bm x-\bm y\|.
\end{equation}
Under the null, set $\bm\delta_\Omega=0$. The quantity $\mathfrak d_T(\Omega)$ is the bounded-Lipschitz distance between the specified laws.

For $M>1+\|\bm\delta_\Omega\|_\infty$, let $L_M=\max_{|z|\le M}g'(z)$ and define
\begin{equation}\label{eq:Cerrorbudget}
 \mathfrak e_T^C(\Omega)=\inf_{\substack{M>1+\|\bm\delta_\Omega\|_\infty\\0<h<1}}
 \left[(2+L_M/h)\mathfrak d_T(\Omega)
 +4J e^{-(M-1-\|\bm\delta_\Omega\|_\infty)^2/2}+2L_Jh\right].
\end{equation}
\begin{lemma}[A distributional error bound]\label{lem:transform}
Under the respective joint-limit conditions,
\begin{equation}\label{eq:Ccdfbound}
 \sup_x\left|\PP(C_T\le x)-\PP\{\mathscr C(\bm\delta_\Omega,\Gam)\le x\}\right|
 \le\mathfrak e_T^C(\Omega)\longrightarrow0.
\end{equation}
\end{lemma}
\begin{proof}
To establish $\mathfrak d_T(\Omega)\to0$, consider the functions in \eqref{eq:BLerror} on $[-M,M]^J$. They admit a finite $3a$-net in the uniform norm.

To construct it, use a finite spatial grid with mesh $a/\sqrt J$ and round function values on the grid to multiples of $a$. There are finitely many value patterns.

Choose one function from each nonempty pattern. Two functions with the same pattern differ by at most $a$ on the grid and by at most $3a$ throughout the cube.

For the resulting $\mathfrak n_{M,a}$ functions $h_1,\ldots,h_{\mathfrak n_{M,a}}$,
\begin{align*}
 \mathfrak d_T(\Omega)
 &\le\max_{\ell\le\mathfrak n_{M,a}}|\E h_\ell(\bm Z_T)-\E h_\ell(\bm Y_\Omega)|+6a\\
 &\quad+2\PP(\|\bm Z_T\|_\infty>M)+2\PP(\|\bm Y_\Omega\|_\infty>M).
 \end{align*}
For fixed $M,a$, the maximum tends to zero by weak convergence and bounded continuity. The vector laws are tight; increasing $M$ and decreasing $a$ prove the claim.

Clip every coordinate to $[-M,M]$ and write
\[
 A_M(\bm z)=\sum_jw_jg\{(-M)\vee(z_j\wedge M)\},\qquad
 |A_M(\bm z)-A_M(\bm z')|\le L_M\|\bm z-\bm z'\|.
\]
For $h>0$, the continuous function
\[
 \theta_{x,h}(u)=\begin{cases}
 1,&u\le x,\\
 (x+h-u)/h,&x<u<x+h,\\
 0,&u\ge x+h
 \end{cases}
\]
lies between the indicators at $x$ and $x+h$. It follows from \eqref{eq:BLerror} that
\[
 |\E\theta_{x,h}\{A_M(\bm Z_T)\}-\E\theta_{x,h}\{A_M(\bm Y_\Omega)\}|
 \le(1+L_M/h)\mathfrak d_T(\Omega).
\]
To control clipping, apply \eqref{eq:BLerror} to $r_M(\bm z)=1\wedge(\|\bm z\|_\infty-M+1)_+$. Then
\begin{align*}
 \PP(\|\bm Z_T\|_\infty>M)
 &\le\E r_M(\bm Z_T)\le\PP(\|\bm Y_\Omega\|_\infty>M-1)+\mathfrak d_T(\Omega),\\
 \PP(\|\bm Y_\Omega\|_\infty>M-1)
 &\le2J e^{-(M-1-\|\bm\delta_\Omega\|_\infty)^2/2}.
 \end{align*}
Combining the two smoothed-indicator inequalities at $x$ and $x-h$ with \eqref{eq:intervalprob} bounds the absolute CDF difference by
\[
 (2+L_M/h)\mathfrak d_T(\Omega)
 +4J e^{-(M-1-\|\bm\delta_\Omega\|_\infty)^2/2}+2L_Jh.
\]
Taking the infimum proves \eqref{eq:Ccdfbound}. For convergence, first fix $M,h$, then let $T\to\infty$, $M\to\infty$, and $h\downarrow0$ in that order.
\end{proof}

\subsection{Right-tail equivalence with a remainder}
Let $J\ge2$, $\rho_* =\max_{i\ne j}\Gamma_{ij}<1$ and $\kappa=(1-\rho_*)/4$. Put $X_j=w_jg(G_j)$ for $\bm G\sim \mathcal{N}_J(\bm0,\Gam)$.
\begin{lemma}[Cauchy right tail]\label{lem:tail}
As $x\to\infty$,
\begin{equation}\label{eq:tailrate}
 \PP\left(\sum_jX_j>x\right)
 =\frac1{\pi x}+O\{x^{-1-\kappa}(1+\log x)^2\}.
\end{equation}
\end{lemma}
\begin{proof}
The uniformity of $\Phi(G_j)$ gives an exact marginal identity. For $y>0$,
\begin{equation*}
 \PP(X_j>y)=\frac1\pi\arctan(w_j/y)
 =\frac{w_j}{\pi y}+r_j(y),\qquad
 |r_j(y)|\le\frac{w_j^3}{3\pi y^3}.
\end{equation*}
The remainder follows from $\arctan u-u=-\int_0^u t^2/(1+t^2)\,dt$.

For sufficiently large $v$, define $q_j(v)>1$ by
\[
 1-\Phi\{q_j(v)\}=\pi^{-1}\arctan(w_j/v).
\]
Since $g$ is odd and increasing, $\{|X_j|>v\}=\{|G_j|>q_j(v)\}$. For constants depending only on the fixed weights,
\[
 c_1/v\le1-\Phi\{q_j(v)\}\le c_2/v.
\]
The upper normal-tail bound gives $q_j(v)^2\le2\log v+C$. Using the lower bound in \eqref{eq:mills},
\[
 \frac{c_2}{v}\ge\frac{\phi\{q_j(v)\}}{2q_j(v)},\qquad
 q_j(v)^2\ge2\log v-2\log q_j(v)-C.
\]
Consequently,
\begin{equation}\label{eq:quantilebound}
 2\log v-\log\log v-C\le q_j(v)^2\le2\log v+C.
\end{equation}
Let $q_*(v)=\min_jq_j(v)$. For signs $e_i,e_j\in\{-1,1\}$, the event $e_iG_i>q_*(v)$, $e_jG_j>q_*(v)$ implies $e_iG_i+e_jG_j>2q_*(v)$.

This Gaussian sum has variance $2+2e_ie_j\Gamma_{ij}\le2(1+\rho_*)$. Lemma~\ref{lem:normal} therefore gives
\[
 \PP\{e_iG_i>q_*(v),e_jG_j>q_*(v)\}
 \le\exp\{-q_*(v)^2/(1+\rho_*)\}.
\]
Summing over the four sign pairs and applying \eqref{eq:quantilebound} yields
\begin{equation*}
 \PP(|X_i|>v,|X_j|>v)
 \le Cv^{-2/(1+\rho_*)}(1+\log v)^2.
\end{equation*}
This bound includes opposite signs and hence controls possible cancellation.

Put $a=(3+\rho_*)/4=1-\kappa$, $v=x^a$ and $b_x=(J-1)v$. Let $B_x$ mean that at least two coordinates satisfy $|X_j|>v$. Then
\[
 \PP(B_x)\le\sum_{i<j}\PP(|X_i|>v,|X_j|>v)
 \le Cx^{-2a/(1+\rho_*)}(1+\log x)^2.
\]
For large $x$, $Jv<x$. Outside $B_x$, at most one coordinate is large, and
\[
 \left\{\sum_jX_j>x\right\}\setminus B_x
 \subseteq\bigcup_j\{X_j>x-b_x\},\qquad
 \left(\bigcup_j\{X_j>x+b_x\}\right)\setminus B_x
 \subseteq\left\{\sum_jX_j>x\right\}.
\]
Pairwise intersections in the second union are contained in $B_x$. Therefore
\begin{equation}\label{eq:tailsandwich}
 \sum_j\PP(X_j>x+b_x)-C\PP(B_x)
 \le\PP\left(\sum_jX_j>x\right)
 \le\sum_j\PP(X_j>x-b_x)+\PP(B_x).
\end{equation}
Because $b_x\le x/2$ eventually,
\[
 \left|\frac1{x\pm b_x}-\frac1x\right|\le\frac{2b_x}{x^2},\qquad
 \sum_j\PP(X_j>x\pm b_x)=\frac1{\pi x}+O(v/x^2+x^{-3}),
\]
where $\sum_jw_j=1$ is used. The total error in \eqref{eq:tailsandwich} is bounded by
\[
 C\{x^{a-2}+x^{-3}+x^{-2a/(1+\rho_*)}(1+\log x)^2\}.
\]
Finally,
\[
 2-a=1+\kappa,\qquad
 \frac{2a}{1+\rho_*}-(1+\kappa)
 =\frac{(1-\rho_*)^2}{4(1+\rho_*)}\ge0.
\]
Every error term has the order asserted in \eqref{eq:tailrate}.
\end{proof}

\subsection{Proofs for the combined test}
\begin{proof}[Proof of Corollary~\ref{cor:Climit}]
Theorems~\ref{thm:null} and \ref{thm:alternative} identify the two Gaussian vector limits. Lemma~\ref{lem:transform} then gives, under the respective hypotheses,
\[
 \sup_x|\PP(C_T\le x)-\PP\{\mathscr C(\bm\delta_\Omega,\Gam)\le x\}|
 \le\mathfrak e_T^C(\Omega)\longrightarrow0.
\]
Lemma~\ref{lem:continuity} gives a continuous limit. The same bound holds for closed upper tails by taking left limits in the CDF inequality. At $x=\cot(\pi\eta)$ it proves \eqref{eq:Cpower}.

Nonnegativity of $\nu$ implies $\delta_\Omega(\lambda_j)\ge0$. Coupling both limits with the same $\bm G$ gives
\[
 \begin{aligned}
 \mathscr C(\bm\delta_\Omega,\Gam)-\mathscr C(\bm0,\Gam)
 &=\sum_jw_j\int_0^{\delta_\Omega(\lambda_j)}g'(G_j+t)\,dt\ge0.
 \end{aligned}
\]
Thus local rejection probability is at least the limiting null rejection probability at a common critical value.
\end{proof}

\begin{proof}[Proof of Theorem~\ref{thm:size}]
Suppose first that $J\ge2$. Lemma~\ref{lem:kernel} ensures $\rho_*<1$. Set $x_\eta=\cot(\pi\eta)$. As $\eta\downarrow0$,
\[
 x_\eta=\frac1{\pi\eta}+O(\eta),\qquad
 \frac1{\pi x_\eta}=\eta+O(\eta^3),\qquad
 x_\eta^{-1-\kappa}(1+\log x_\eta)^2
 \le C\eta^{1+\kappa}(1+|\log\eta|)^2.
\]
These expansions follow by dividing the Taylor formulas for $\cos u$ and $\sin u$ at $u=\pi\eta$. Lemma~\ref{lem:tail} yields
\begin{equation*}
 \left|\PP\{\mathscr C(\bm0,\Gam)>x_\eta\}-\eta\right|
 \le C\eta^{1+\kappa}(1+|\log\eta|)^2.
\end{equation*}
Equation~\eqref{eq:Ccdfbound} and the rejection event in \eqref{eq:Cauchy} give
\begin{equation*}
 \left|\PP_0(p_{C,T}\le\eta)-\eta\right|
 \le\mathfrak e_T^C(0)+C\eta^{1+\kappa}(1+|\log\eta|)^2.
\end{equation*}
For each fixed $\eta>0$, $\mathfrak e_T^C(0)\to0$. Dividing by $\eta$ after taking $T\to\infty$ proves \eqref{eq:sizebound}; then $\eta\downarrow0$ proves \eqref{eq:size}.

For a simultaneous sequence $\eta_T\downarrow0$, the bound in \eqref{eq:Cerrorbudget} requires $\mathfrak e_T^C(0)/\eta_T\to0$ in addition to the joint weak limit.

For $J=1$, $w_1=1$ and
\[
 \frac12-\frac1\pi\arctan g(G_1)
 =\frac12-\{\Phi(G_1)-1/2\}=1-\Phi(G_1)\sim\operatorname{Unif}(0,1).
\]
The limiting analytic rejection probability is exactly $\eta$ at each fixed level, proving the remaining statement.
\end{proof}

\section{Implementation identities for the simulations}\label{app:simulation}
\subsection{Residual-degrees-of-freedom normalization}\label{app:sim_df}
For the covariance estimator $\SSS$ in \eqref{eq:cov}, let $n=T-p-1$ and let $s_1,\ldots,s_n$ be its positive eigenvalues padded with zeros when necessary. The simulated marginal correction is
\begin{equation}\label{eq:sim_corrections}
	\begin{gathered}
		\omega_\ell(\lambda)
		= \frac{s_\ell}{s_\ell+\lambda},
		\qquad
		\widehat h_{n,T}(\lambda)
		= \frac{1}{n}\sum_{\ell=1}^n \omega_\ell(\lambda),
		\qquad
		\widehat b_{n,T}(\lambda)
		= 1-\widehat h_{n,T}(\lambda),
		\\[0.6em]
		\widehat\xi_{n,T}(\lambda)
		= \frac{\widehat h_{n,T}(\lambda)}
		{\widehat b_{n,T}(\lambda)},
		\qquad
		\widehat\Delta_{n,T}(\lambda)
		= \frac{2}
		{n\,\widehat b_{n,T}(\lambda)^4}
		\sum_{\ell=1}^n
		\bigl\{\omega_\ell(\lambda)
		-\widehat h_{n,T}(\lambda)\bigr\}^2 .
	\end{gathered}
\end{equation}

The statistic and marginal $p$-value used in the reported results are
\begin{equation}\label{eq:sim_statistic}
 Z_T^{\mathrm{df}}(\lambda)=
 \frac{\sqrt n\{W_T(\lambda)-\widehat\xi_{n,T}(\lambda)\}}
      {\sqrt{\widehat\Delta_{n,T}(\lambda)}},
 \qquad
 p_T^{\mathrm{df}}(\lambda)=1-\Phi\{Z_T^{\mathrm{df}}(\lambda)\}.
\end{equation}
Here $W_T(\lambda)$ is defined in \eqref{eq:W}. Define $C_T^{\mathrm{df}}$ by \eqref{eq:Cauchy}, replacing each $p_{j,T}$ by $p_T^{\mathrm{df}}(\lambda_{j,T})$. The experiments use the five scaled ridges in \eqref{eq:sim_scale}.

With $h=h_T(\lambda)$ and $v=v_T(\lambda)$ from \eqref{eq:marginalmoments}, the normalization changes these moments to $h_n=(T/n)h$ and $v_n=(T/n)v$. The rational bounds in Lemma~\ref{lem:normalization} give
\[
 \max_j\{|\widehat\xi_{n,T}(\lambda_j)-\widehat\xi_T(\lambda_j)|
       +|\widehat\Delta_{n,T}(\lambda_j)-\widehat\Delta_T(\lambda_j)|\}
 \le C/T
\]
on events whose probabilities tend to one. The exact relation
\[
 Z_T^{\mathrm{df}}(\lambda)
 =\sqrt{\frac nT}\sqrt{\frac{\widehat\Delta_T(\lambda)}
                                  {\widehat\Delta_{n,T}(\lambda)}}Z_T(\lambda)
 +\frac{\sqrt n\{\widehat\xi_T(\lambda)-\widehat\xi_{n,T}(\lambda)\}}
             {\sqrt{\widehat\Delta_{n,T}(\lambda)}}
\]
then gives $\max_j|Z_T^{\mathrm{df}}(\lambda_{j,T})-Z_T(\lambda_{j,T})|=O_p(T^{-1/2})$ on the grid in Corollary~\ref{cor:scale}. Since $g$ has a bounded derivative on every compact interval and both finite vectors are tight, $C_T^{\mathrm{df}}-C_T=O_p(T^{-1/2})$; hence the Cauchy limits and relative size conclusion are unchanged.

\begin{samepage}
\subsection{The leave-one-out comparator}\label{app:sim_loo}
This subsection specifies the simulation implementation of the selector in \eqref{eq:loo_rule}. The experiment replaces the original multiplier grid by $\mathcal A$ in \eqref{eq:sim_scale} and uses the common residual-degrees-of-freedom variance estimate:
\[
 \widehat a_{\mathrm L}\in\argmax_{a\in\mathcal A}
 \frac{\widetilde H_T(a\widehat s_T^{\mathrm{raw}})}
      {\sqrt{\widehat\Delta_{n,T}(a\widehat s_T^{\mathrm{raw}})}},
 \qquad \widehat\lambda_{\mathrm L,T}=\widehat a_{\mathrm L}\widehat s_T^{\mathrm{raw}}.
\]
Ties are resolved by the smallest multiplier. LOO rejects when $p_T^{\mathrm{df}}(\widehat\lambda_{\mathrm L,T})<\eta$, using \eqref{eq:sim_statistic}. The uncentered Gram matrix and candidate penalties retain their $T$-based scale. The regression and ridge scale are computed from the full sample, and no additional correction for selection is applied.

\end{samepage}

\subsection{The population-formula benchmark}\label{app:sim_formula}
For the formula optimum, the spectral law is uniform on $[1,2]$ and the generating allocation is $\vartheta(\tau)=1+\tau/2$. To express the criterion in residual-degrees-of-freedom covariance units, put $c_n=N/n$ and define, for $y>0$,
\begin{align*}
 \mathfrak b_n(y)&=1-c_n\int_1^2\frac{\tau}{y+\tau}\,d\tau,\\
 \mathfrak h_n(y)&=c_n\int_1^2\frac{\tau^2}{(y+\tau)^2}\,d\tau,
 \qquad
 m_0(y)=\int_1^2\frac{1+\tau/2}{y+\tau}\,d\tau.
 \end{align*}
The effective ridge $y$ corresponds to the nominal ridge $\ell=y\mathfrak b_n(y)$ for $\overline{\mathbf S}_T=(T/n)\SSS$. Over
\[
 \mathcal D_n=\{y>0:\mathfrak b_n(y)>0,\quad
             0<\mathfrak h_n(y)<1,\quad
             y\mathfrak b_n(y)\in[1.5\times10^{-5},150]\},
\]
choose
\begin{equation}\label{eq:sim_formula}
 y^\star\in\argmax_{y\in\mathcal D_n}
       m_0(y)\sqrt{\frac{1-\mathfrak h_n(y)}{2\mathfrak h_n(y)}},
 \qquad
 \lambda_{\mathrm F,T}^{\star}=\frac nT y^\star\mathfrak b_n(y^\star).
\end{equation}
The penalty $\lambda_{\mathrm F,T}^{\star}$ is expressed directly in the covariance units of $\SSS$. The factor $n/T$ converts the ridge from $\overline{\mathbf S}_T$ to $\SSS$; the common factor $d^{-1}$ in the population power objective does not affect its maximizer.

\end{document}